\documentclass[11pt]{article}
\usepackage{fullpage,xspace,times}
\usepackage{graphicx}
\usepackage{amsmath,upgreek}
\usepackage{amsfonts}
\usepackage{amssymb}
\usepackage{epstopdf}

\newtheorem{theorem}{Theorem}

\newtheorem{lemma}[theorem]{Lemma}

\newtheorem{proposition}[theorem]{Proposition}

\newcommand{\Xomit}[1]{ }
\newenvironment{proof}[1][Proof]{\textbf{#1.} }{\ \rule{0.5em}{0.5em}}

\mathchardef\mhyphen="2D

\newcommand{\eps}{\upvarepsilon}

\def\eps{\varepsilon}

\begin{document}

%\pagestyle{myheadings}

%\oddsidemargin -0.99cm
%\evensidemargin -0.99cm

%\title{The price of clustering, batched algorithms, and greedy algorithms for several variants of bin packing}
%\thanks{This work was partially supported by a grant from GIF - the German-Israeli Foundation for Scientific Research and Development (grant number I-1366-407.6/2016).}}

\title{Open-end bin packing: new and old analysis approaches}

\date{}

\author{Leah Epstein\thanks{
Department of Mathematics, University of Haifa, Haifa, Israel.
\texttt{lea@math.haifa.ac.il}. }}

\maketitle

%TODO mandatory: add short abstract of the document
\begin{abstract}
We analyze a recently introduced concept, called the price of clustering, for variants of bin packing called open-end bin packing problems (OEBP). Input items have sizes, and they also belong to a certain number of types. The new concept deals with the comparison of optimal solutions for the cases where items of distinct types can and cannot be packed together, respectively. The problem is related to greedy bin packing algorithms and to batched bin packing, and we discuss some of those concepts as well. We analyze max-OEBP, where a packed bin is valid if by excluding its largest item, the total size of items is below $1$. For this variant, we study the case of general item sizes, and the parametric case with bounded item sizes, which shows the effect of small items. Finally, we briefly discuss min-OEBP, where a bin is valid if the total size of its items excluding the smallest item is below $1$, which is known to be an entirely different problem.
\end{abstract}

\section{Introduction}
We study a variant of bin packing, called {\it Open End Bin Packing} (OEBP). The input for one-dimensional bin packing problems consists of items with rational sizes in $(0,1]$. There may be additional attributes for some variants and for certain analysis approaches. In bin packing problems, the goal is to partition the items into subsets called bins, under conditions on the contents of a bin. In the classic or standard bin packing problem \cite{J74,JoDUGG74}, a bin is a set of items whose total size is at most $1$. In open end bin packing problems \cite{Zhang98,LDY01,YangL03,EL08,GZ09,LYX10A,LYX10T,BEL20,EL20},  a bin is a set of items such that by removing one item, the total size is below $1$. In Max-OEBP, the removed item is the largest item, and in Min-OEBP the removed item is the smallest one. We are interested in greedy algorithms for these two problems. Such algorithms may be offline algorithms, which view the input items as a set, and online algorithms, which receive items one by one, and pack every arriving item before seeing the next one.

\noindent{\bf Measures.}
An approximation algorithm has an asymptotic approximation ratio of at most $R$, if there
exists a constant $C\geq 0$ (which is independent of the input),
such that for any input $I$, the cost of the algorithm for $I$
does not exceed the following value: $R$ times the optimal cost
for this input plus $C$. The asymptotic approximation ratio is the infimum $R$ for which the inequality holds for every input.
If $C=0$, then the approximation ratio is
called absolute. A specific optimal (offline) algorithm as well as its cost are denoted by
$OPT$ ($OPT(I)$ is sometimes used for the cost for a fixed input $I$). An alternative
definition of the asymptotic approximation ratio is the supreme
limit of the ratio between the two costs of the algorithm and of $OPT$, as a
function of the latter cost, taking the maximum or supremum over
the inputs with the same optimal cost. If the algorithm is online, the term {\it approximation ratio} is sometimes replaced by the term {\it competitive ratio}, but the definition is unchanged, and we will use the term approximation ratio uniformly.
For bin packing problems and approximation algorithms for them, the asymptotic measures are considered to be more meaningful and therefore we will mostly discuss asymptotic measures. When we do not write if the measure is asymptotic or absolute, the results hold for both measures.

%%{\bf Batched bin packing and bin packing with clustering.}
\noindent{\bf Bin packing with item types.}
We describe two models for bin packing problems, where the input is presented together with a fixed partition into subsets. This partition is a part of the input, and not the partition into bins, which is required from an algorithm. The subsets of the partition of the input are also called item types.
Such models may be offline or online. The first variant out of such problems, called  batched bin packing, was defined as a model related to online algorithms. This is an intermediate model which has features of both the extremes of the offline problem and the online problem. It was studied for combinatorial optimization problems, including classic bin packing \cite{GJYbatch,BBDGT,Do15,Epstein16}, and here we define it for bin packing problems. Given an integer parameter $\ell$, items are
presented in $\ell$ batches, for an integer $\ell \geq 1$. For each
batch, the algorithm receives all its items at once, and these
items are to be packed irrevocably before the next batch is
presented (or before termination, if the current batch is not the last one).
Two sub-models have been defined. The first one is where the algorithm may use bins from previous batches for the packing, in which case the crucial difficulty that the algorithm is faced with is that it does not see future batches  when a batch is packed. In the second one, which we study here, different batches cannot share bins. We stress that the two versions are semi-online models, in the sense that the algorithm sees the current batch and past batches, but it does not know anything about future batches. In particular, such batches may be empty. However, the second model, which does not allow combining items from different batches in bins is equivalent to an offline model because the different batches are packed independently. It is known that for $\ell=2$, the tight bound on the asymptotic approximation ratio (for the second sub-model and classic bin packing) is $1.5$ \cite{Epstein16}, and for larger values of $\ell$ the ratio increases and grows to approximately $1.69103$. In principle, it is possible to find an optimal solution for every batch, but in some cases a different algorithm is easier to analyze, and still tight bounds are found.

The last model leads us to {\it bin packing with clustering}, which is an offline problem, and it is related to the variant of batched bin packing studied here (where items of different batches cannot share bins).
In this problem, every item has
a second attribute, called a cluster index or a color. A global
solution is one where items are packed without considering their
clusters, i.e., it is just a solution of the corresponding bin packing
problem for this input. A clustered solution is one where every
cluster or color must have its own set of bins, and items of
different clusters cannot be packed into a common bin.  This corresponds to cloud computing, where an input data stream is split quickly into sub-inputs, which are considered independently \cite{AESV}.
Obviously, in order to avoid
degenerate cases (for example, the case that every item has its own cluster), an assumption on the input is enforced. It is assumed that every cluster is sufficiently
large, and an optimal solution for each cluster has at least a fixed number of bins. Here, we assume that every cluster requires at least two bins, so that (for example) a situation where every item belongs to its own cluster will not be encountered.
This model was introduced by Azar et al.~\cite{AESV} (see also \cite{E20}) for classic bin packing. The goal is to compare optimal solutions, that
is, to compare an optimal clustered solution to an optimal global
solution, also called a globally optimal solution. We are
interested in the worst-case ratio over all valid inputs, and this
(absolute and asymptotic) approximation ratio is called {\it price of clustering} (PoC). From an algorithmic
point of view, the goal is to design an approximation algorithm
for which it is not allowed to mix items of different clusters,
while the algorithm still has a good approximation ratio compared
to a globally optimal solution. Our results will hold for both the absolute and the asymptotic approximation ratio, and thus we just use the term PoC.

The last two models are related, though in batched bin packing (with separate bins for the different batches), the number of batches is the analyzed parameter. That is, there is a small number of clusters which are typically large. In bin packing with clustering, there is typically a large number of clusters, which may be relatively small, and a lower bound on their optimal numbers of bins can be seen as a parameter.
We study OEBP with respect to greedy algorithms and packing with item types.
We mostly focus on Max-OEBP, and provide some observations for Min-OEBP as well.  Note that we analyze Max-OEBP with respect to multiple aspects, and in particular we consider greedy algorithms.
As mentioned above, in order for the results to be comparable with other results for bin packing, our analysis is for the asymptotic approximation ratio. For bin packing with clustering, previous work is for the absolute approximation ratio and our results are valid for both measures.

In the general variants of the problems it is assumed that the size of an item may be very close to the capacity of a bin, or even equal to it. However, in many applications bin capacities are larger than item sizes. The parametric case is the one where sizes do not exceed a fixed value $\beta$, where $0<\beta\leq 1$. For some problems, adding the size constraint makes the problem different, and more interesting. There are models where  the interesting cases are only those where $\beta$ is a reciprocal of an integer, but in other cases there are additional values of $\beta$ that should be considered.
For classic bin packing, the asymptotic approximation ratio (for essentially all reasonable algorithms) tends to $1$ as $\beta$ tends to zero. However, there are variants where the effect of small items does not decrease the ratio that much \cite{ACFR16}. The problems studied here also have the property that the asymptotic approximation ratio tends to $1$ for small $\beta$, but the price of clustering is never close to $1$ (this is the case also for classic bin packing).

For the PoC, we assume that every cluster requires at least two bins. The resulting constraint on the items of a cluster for this assumption is slightly different for distinct variants of bin packing, but we feel that for comparison, the constraint should be defined in a similar way based on optimal solutions (for the appropriate variants). Here, we show that the PoC is equal to $3$ for max-OEBP and to $4$ for min-OEBP. We also find tight bounds for the general parametric case of items of sizes in $(0,\beta]$ and any $0<\beta<1$, which requires a careful analysis, and we shed some light on the parametric case for min-OEBP. For small values of $\beta$, the PoC tends to $2$ for both variants.
The results of \cite{AESV} for classic bin packing are based on the assumption that the cost of an optimal solution for every cluster is at least $k$ for an integer parameter $k\geq 2$. In the case $k=2$, they show that the PoC is exactly $2$, but for $k=3$, it is strictly below $2$, and more
specifically, it is at most $1.951$ (but at least $1.93344$). These bounds were slightly improved for $k=3$ in \cite{E20} to $1.93667$ and $1.93558$, respectively, where the case of larger $k$ is studied as well. The example of \cite{AESV} and $k=2$ holds only for $\beta=1$, but it is not difficult to see that $2$ is the tight bound for $k=2$ and any $\beta<1$, by using clusters with $N+1$ jobs of size $\frac 1N$ each for a large integer $N$. We find here that unlike classic bin packing, the PoC is much smaller for max-OEBP and for min-OEBP if all items are small.

%Several versions of OEBP problems were studied. The first version is an unfair online variant where an online algorithm has to pack the items in the order that they are presented while an optimal offline solution can reorder the items. The items of a bin are tested in the order of assignment, and a bin is valid if its items excluding the last one have total size strictly below $1$.

\noindent{\bf Literature review of algorithmic results.}
Greedy algorithms for classic bin packing were studied already in the early stages of the study of such optimization problems \cite{JoDUGG74,J74}.
The simplest algorithm for classic bin packing is Next Fit (NF). This algorithm assigns the first item into an empty bin, and this bin becomes active. At every time, there will be one active bin. For each additional item, it is packed into the active bin if possible, and otherwise a new bin is defined as active and the new item is packed into it. A class of common greedy algorithms is called Any Fit (AF). These algorithms never pack a new item into an empty bin, if there is another possible way to pack the item. This class contains algorithms that choose a bin arbitrarily (out of the existing bins that can accommodate the item), but also  First Fit (FF), that selects the bin of minimum index (where the item can be packed), and  Worst Fit (WF), that selects the bin of smallest load or  maximum free space (if the item can be packed there). There are other intuitive choices that an algorithm may select. There are also offline versions of these algorithms, where items are sorted by non-increasing size. The word {\it Decreasing} is added to the names of the algorithms, and the names of the algorithms are abbreviated by NFD, AFD, FFD, and WFD. These algorithms were defined approximately fifty years ago \cite{J74,JoDUGG74}.
It is known that the asymptotic approximation ratio of FF and BF is $1.7$ \cite{JoDUGG74} (in fact, the absolute approximation ratio is also $1.7$ \cite{DS12}). For WF (and AF in general) and NF, the ratio is $2$ \cite{J74} (there are some other special cases of AF, for which the ratio is equal to $1.7$).
The asymptotic approximation ratios of the sorted variants are smaller, for example, the ratio for FFD is $\frac{11}9$ \cite{J74,DLHT}. For NFD, the asymptotic approximation ratio is almost $1.7$ \cite{BC81}, and for AFD algorithms in general, the asymptotic approximation ratio is at most $1.25$ \cite{J74,JoDUGG74}.
Most parametric results are for classic bin packing.

Greedy algorithms for classic bin packing were also analyzed with respect to the parametric case. For example, the asymptotic approximation ratio of FF if all items have sizes in $(0,\frac 1t]$ for an integer $t\geq 2$ is $\frac{t+1}t$ \cite{JoDUGG74}, and the cases where the parameter $\beta$ is not a reciprocal of an integer are identical to the cases where it is (for example, the cases $\beta=0.4$ and $\beta=0.5$ are equivalent). For NF and WF, the asymptotic competitive ratio is $2$ for $\beta>\frac 12$. For reciprocals of integers, it is $\frac{t}{t-1}$, and more generally, for $\beta\leq \frac 12$, the ratio is exactly $\frac{1}{1-\beta}$ \cite{J74}. There is also work on FFD and NFD in the parametric case \cite{J74,Cs93,Xu00,BC81}.
Since greedy algorithms without sorting are greedy algorithms, we mention that other online algorithms are being studied continuously  during the last fifty years \cite{BBDEL_ESA18,BBDEL_newlb}.
The
articles of van Vliet \cite{Vliet92} and Balogh,
B{\'e}k{\'e}si, and Galambos \cite{BBG} contain lower bounds on the asymptotic approximation ratio of online algorithms, where in addition to the
lower bounds for the standard case also lower bounds for the parametric
case, and the methods for designing algorithms  \cite{LeeLee85,RaBrLL89}
can also be used for parametric versions.

Four variants of OEBP were studied in the past. The first two (which are not those studied here) are online problems, where items are packed into a bin in the order of arrival, and the total size for a valid bin excluding the item that arrived last has to be below $1$.
In the unfair variant \cite{LDY01,Zhang98}, the optimal offline solution can reorder the items, so the last item of a bin is not necessarily the one that arrived last. In the fair variant \cite{YangL03,BEL20}, all algorithms have to obey the order of arrival.
The unfair variant was studied with respect to a certain class of algorithms for the parametric case \cite{Zhang98}. Since an offline algorithm can reorder while the online algorithm cannot reorder, the bounds are based on the simple property (used also here) that any bin of an offline solution has items whose total size is below $1+\beta$, while an online algorithm has items of total size above $1$ packed into almost every bin. This upper bound holds for many algorithms for other variants of OEBP.
For the fair online variant, Yang and Leung \cite{LDY01} proved lower bounds on the asymptotic approximation ratio (which were slightly improved later \cite{BEL20}). In particular, they showed a lower bound of $1.630297$. An online algorithm of asymptotic approximation ratio $1.691561$ was designed \cite{BEL20}. For the case where there are no items of size $1$, the lower bound of  \cite{LDY01} is much smaller ($1.415715$), and the results of Zhang \cite{Zhang98} imply that there is a simple online algorithm whose asymptotic approximation ratio is $1.5$ for this case. For this case, the algorithm of \cite{BEL20} has an asymptotic approximation ratio of at most $1.44465$.

Max-OEBP was not studied as an online problem, but algorithms similar to those of the ordered problem yield the same bounds. Offline algorithms were studied and in particular, FFD was studied for the case with items of sizes strictly below $1$, and its asymptotic approximation ratio is $1.5$ \cite{theFFD}.
As for Min-OEBP, greedy and online algorithms were analyzed for this variant
\cite{LYX10T}. In particular, several algorithms including  FFD and First Fit
Increasing (FFI) have asymptotic approximation ratios of $\frac{71}{60}\approx
1.18333$, for NF, WF, and FF these ratios are $4$, $3$, and $2$, respectively. An improved online algorithm was designed recently \cite{EL20}. The paper \cite{LYX10T} also contains a partial parametric analysis.

There are several motivations of studying greedy algorithms. These algorithms are simple and natural, easy to implement, and they are frequently found in the industry. We analyze them in the same work as bin packing with item types since greedy algorithms are useful in the analysis of the PoC and batched bin packing.
Definitions for online algorithms and greedy algorithms for OEBP problems should carefully take into account the property that packed bins must remain valid throughout the execution of an algorithm.

%\medskip\medskip

\noindent{\bf Notation.}
For a fixed input, let $\ell$ be the number of clusters. Let
$OPT_i$ be the number of bins in an optimal solution for the $i$th
cluster, whose input is $I_i$. We let $I$ be the set of items
$I=\bigcup_{1 \leq i \leq \ell} I_i$, where $n=|I|$. Let $OPT$ be
a globally optimal solution for $I$, and its cost is denoted by $OPT(I)$, that is, we use the notation of an optimal solution for the globally optimal solution.
We let $A_i$ be the number of bins in the output of an algorithm for cluster $i$, which may be an optimal solution or any other solution. For batched bin packing, we have batches rather than clusters, but the notation defined here is unchanged.
For an algorithm ALG and an input $I$, the cost of ALG, which is the number of bins it creates for $I$ is denoted by $ALG(I)$.
In some cases, We will use weights for the analysis of the upper bound. Weights allow us to compare two
solutions, using the fact that the total
weight of all input items is consistent. For any weight function $w$ defined in this work, we will use $W$ to denote the total weight of all items of $I$, and $W_j$ denotes the total weight for $I_J$.

\noindent{\bf Structure of the paper and results.}
The results for the PoC appear after we deal with the analysis of greedy algorithms. This is done not only because greedy algorithms are natural and interesting in their own right, but also because greedy algorithms are used for the analysis for batched bin packing and bin packing with clustering. While the model assumes that every batch or cluster can be packed using an optimal packing, it is often the case that using a greedy algorithm for the analysis instead still allows us to find tight bounds.

In Section \ref{greed} we analyze greedy algorithms for max-OEBP. In particular, in Section \ref{wosort} we analyze algorithms that do not sort the input. We find tight bounds on the asymptotic approximation ratios of NF, WF, and FF. The main result of this part is the tight analysis for FF as a function of $\beta$. While the tight bound for NF (and WF) is $1+\beta$, for FF this bound is tight only for values of $\beta$ that are reciprocals of integers, and for other values it is smaller.  In Section \ref{wsort}, we analyze greedy approaches for sorted inputs, and we study NFD and FFD. Interestingly, NFD also performs better exactly for the same values of $\beta$ as FF, but the effect of sorting is not large. For all these algorithms, the overall bound (that is, the bound for $\beta=1$) is $2$.
For FFD, the overall bound is smaller and equal to $\frac 32$. (this was known for the case where there are no items of size $1$ in the input \cite{theFFD}).
Section \ref{ittypes} is dedicated to inputs with item types, for max-OEBP. We show that for any number of batches $\ell \geq 2$, the tight asymptotic approximation ratio is $2$, and we prove tight bounds for any $\beta$. In Section \ref{pocc} we analyze the price of clustering. The main contribution of this part is a careful analysis of the PoC as a function of $\beta$. The value of the PoC is between $2$ and $3$, depending on the value of $\beta$.
In Section \ref{mino} we discuss the PoC for min-OEBP. We prove the general case, where the value of the PoC is $4$, analyze particular cases, and we provide examples explaining why we do not perform a complete analysis. Omitted proofs can be found in the Appendix.

%\section{Open end bin packing}

%\section{Max-OEBP}

\section{Greedy algorithms for Max-OEBP}\label{greed}
We start the analysis of Max-OEBP with greedy algorithms, first those that do not sort the input and therefore they are in fact online algorithms, and afterwards those that sort the input and apply a greedy algorithm for it. The last class of algorithms can be seen as semi-online algorithms where the input arrives in a  sorted order.
We discuss several common and natural greedy algorithms. In most cases we discuss not only general inputs but also parametric inputs. We usually use the parameter $\beta$ as an upper bound on the sizes of items, that is, item sizes are in $(0,\beta]$.  The algorithms were originally defined for classic bin packing, and here we stress the differences in the definitions of these algorithms for the studied problem.

\subsection{Algorithms without sorting}\label{wosort}
The algorithm First Fit (FF) for this variant is defined as follows. Every item in the list is packed into the bin of the minimum index such that
the item can be added to this bin. An item can be added to a bin in the case where it will remain valid after adding the item, which means that it will satisfy the properties of a bin.  The required property for adding an item is that after adding this item, it is still the case that by removing the largest item of the bin, the total size of items will be smaller than $1$. If there is no such bin, a new bin is used, which is always possible. Similarly, one can apply AF (and in particular WF) or NF. AF algorithms also select a bin out of those that can receive the item (based on the same property), and uses a new bin if there is no such bin. There are two variants of WF, with respect to the choice of a bin out of bins where an assignment will result in a feasible bin.
In the first variant, a new item is added to the existing bin where the current total size is the smallest (and the resulting total size will also be the smallest), and in the second variant, a new item is added to the bin where after the assignment, the total size excluding the largest item will be the smallest.

Note that it is possible that an item is added to a bin whose total size is above $1$. Consider a bin with three items of size $0.35$. If an item of size $0.25$ is added, the bin remains valid. On the other hand, for all these algorithms, any of the outputs will satisfy the property that every bin, possibly excluding the last one, has a total size of items of at least $1$. This holds since a bin whose total size of items is below $1$ can receive an additional item of any size. We will see later that this property does not hold for min-OEBP. For every valid solution, the total size of items for every bin excluding one item is below $1$, and thus, if items sizes are bounded from above by a parameter $\beta$, bins have total sizes below $1+\beta$ (and this holds for min-OEBP as well).

We start with the analysis of NF and WF. In the case of Max-OEBP, NF and both versions of WF have identical behaviours with respect to the asymptotic approximation ratio. Recall that we analyze the case where item sizes are in $(0,\beta]$ for $0<\beta\leq 1$, and therefore the analysis covers both the general case ($\beta=1$) and the parametric case ($\beta<1$). The same upper bound construction is also valid for FF, but it is not always tight as we will see later.

\begin{theorem}\label{firstt}
The asymptotic approximation ratio of NF, of both variants of WF and of AF in general with the parameter $\beta$ is exactly $1+\beta$. It is equal to $2$ for the general case $\beta=1$.
\end{theorem}
Note that particular variants of AF may have a smaller asymptotic approximation ratio, and we show that for FF and an infinite number of values of $\beta$ this is indeed the case.
%\subsection{Proof of Theorem \ref{firstt}}

\noindent\begin{proof}
The upper bound follows from the property that any bin of the output, possibly excluding the last one, has a total size of items of at least $1$, while no bin can have a total size of items of $1+\beta$ or more. This last property holds (as explained above) since the largest item of any bin has size at most $\beta$ and the remaining items have a total size below $1$. This is valid for NF and any AF algorithm.

To prove the lower bound for NF and the two variants of WF (which will imply the lower bound for AF in general), let $M>2N$ be a large positive integer such that $M>\frac 1{\beta}$, and let $\eps=\frac{1-\beta}M$, where $\eps<\beta$.
For a large positive integer $N$, there are $N$ large items, each of size $\beta$, and $N(M+1)$ small items, each of size $\eps$.
The input arrives in the following order. Every large item is followed by exactly $M+1$ small items. The total size of every $M+1$ small items is $(M+1)\cdot \eps=\frac{M(1-\beta)}M+\eps=1-\beta+\eps<1$, so a large item can be packed into a bin with $M+1$ small items. However, the total size together with the large item is above $1$ (it is equal to $1+\eps$), so after $M+2$ items are packed and the next large items arrives, a new bin is opened. For NF, previous bins cannot be used again by definition, so it is clear that the process is repeated and $N$ bins with an identical packing are created. For WF, previous bins can be used, but we claim that the resulting packing is the same as for NF. Assume that some number of bins with one large item and $M+1$ small items were created. A large item cannot be added to such a bin, while a small item can still be added if $\eps$ is sufficiently small. Thus, a new large item is packed into a new bin, and then all bins are considered for any small item. However, while small items are being packed (there are $M+1$ such items), the new bin has the smallest total size of items. All bins have one large item each, so the new bin also has the smallest total size excluding the large item. Thus, all $M+1$ small items are packed into the new bin, as for NF. Thus, all the considered algorithms pack exactly $N$ bins.

An optimal solution acts as follows. First, it creates $\lfloor \frac{\beta}{1+\beta} \cdot N\rfloor \geq \frac{\beta}{1+\beta}\cdot N$ bins with two large items each.
Note that $2\cdot \lfloor \frac{\beta}{1+\beta}\cdot N \rfloor \leq \frac{2N\beta}{1+\beta} < N$ by $\beta<1$. The remaining large items are packed one per bin, and the number of such bins is $N-2\lfloor \frac{\beta}{1+\beta} \cdot N\rfloor \geq \frac{N(1-\beta)}{1+\beta}$. This number is non-negative since $\beta<1$.

The number of bins with at least one large item is $N-\lfloor \frac{\beta}{1+\beta}\cdot N \rfloor$, which is at least $N -\frac{\beta\cdot N}{1+\beta}=\frac{N}{1+\beta}$ and at most $N -\frac{\beta\cdot N}{1+\beta}+1=\frac{N}{1+\beta}+1$, and we show that all small items can be packed into these bins and one additional bin.
Any bin with one large item receives $\lceil \frac M{1-\beta}\rceil -1$ small items, where $(\lceil \frac M{1-\beta}\rceil -1)\eps<\frac M{1-\beta} \cdot \eps=1$ and  $\lceil \frac M{1-\beta}\rceil -1\geq \frac M{1-\beta}-1$. Any bin with two large items receives $M-1$ small items, where $(M-1)\eps=1-\beta-\eps$. The new bin can receive $\lceil \frac M{1-\beta}\rceil >M $ small items.

The number of small items packed with large items is at least $(\frac{M}{1-\beta}-1)\frac{1-\beta}{1+\beta}N+(M-1)\frac{\beta}{1+\beta}\cdot N=M\cdot N-\frac{N}{1+\beta}$. There are at most $N+\frac{N}{1+\beta}<2N<M$ remaining small items, so they can be packed into one bin. The total number of bins for the described packing is at most $\frac{N}{1+\beta}+2$, while the algorithm uses $N$ bins, and the lower bound on the asymptotic approximation ratio follows by letting $N$ grow to infinity.
\end{proof}

\medskip

We turn our attention to better AF algorithms and in particular FF. We found the asymptotic approximation ratio for AF algorithms in general, but some AF algorithms may have a better performance. We will show that for certain values of $\beta$ all AF algorithms including FF still have the same performance, and later we see that for other values of $\beta$, FF has a better performance.
Let $t=\lceil \frac 1{\beta} \rceil-1$, i.e., $t=0$ if $\beta=1$, and for $\beta \in [\frac 1{t+1},\frac 1t)$ for an integer $t$.

\begin{lemma}\label{le2}
For any integer $t \geq 1$, the asymptotic approximation ratio of any AF algorithm is at least $1+\frac{1}{t+1}$.
In the case $\beta=\frac 1{t+1}$ for an integer $t$, the asymptotic approximation ratio of FF and any AF algorithm is exactly $1+\beta$, and in particular, in the general case it is equal to $2$.
\end{lemma}
%\subsection{Proof of Lemma \ref{le2}}
\begin{proof}
As we saw in the previous theorem, the upper bound of $1+\beta$ holds for all AF algorithms including FF.

Consider the following inputs. Let $M,N$ be large positive integers, and let $\eps=\frac{1}{M}$. There are $(t+1)\cdot N\cdot M$ large items of size $\frac{1}{t+1}$, and $(t+1)\cdot N\cdot M \cdot (M-1)$ small items, whose sizes are equal to $\eps$. An optimal solution has $(t+1)\cdot N\cdot M$ bins, each containing $M-1$ small items and one large item. Such bins are feasible since $(M-1)\cdot \eps =\frac{M-1}M<1$.

We prove the lower bound for any AF algorithm, including FF. Assume that an AF algorithm receives the small items before the large items. In this case, the small items are packed into $(t+1)\cdot N\cdot (M-1)$ bins (with $M$ items each), which cannot receive any further items (because there are no smaller items in the input), and $N \cdot M$ bins with large $t+1$ large items each. The total number of bins is $(t+1)\cdot N\cdot (M-1)+N\cdot M=(t+2)\cdot N\cdot M-(t+1)\cdot N$.
The asymptotic approximation ratio is at least $\frac{t+2}{t+1}-\frac 1M$, tending to  $1+\frac1{t+1}$ for $M$ growing to infinity.
\end{proof}

%We will show that in general the asymptotic approximation ratio is $2$ as for NF, but in the case where $\beta\in[\frac{1}{t+1},\frac 1t)$, the asymptotic approximation ratio is $1+\frac{1}{t+1}$. That is, in cases where $\frac 1{\beta}$ is an integer, the ratios are equal, but in other cases the asymptotic approximation ratio for AF is smaller than that of NF.
%\medskip

%\medskip
The last simple construction is not the best possible for all values of $\beta$.
We saw that the asymptotic approximation ratio of FF for max-OEBP
is $1+\frac{1}{\beta}$ for cases that the reciprocal of $\beta$ is
an integer, and the above analysis is tight for those cases. For the general case $\beta=1$ we got tight bounds of
$2$. To summarize, we got tight bounds of
$1+\beta=1+\frac{1}{t+1}$ for $\beta=\frac{1}{t+1}$.

In what follows we analyze FF for max-OEBP in the parametric case where $\beta<1$.
We will show that the asymptotic approximation ratio as a function of $\beta$ is:

$$\displaystyle{R_1(\beta)=\begin{cases} \vspace{0.3cm}
 1+\frac{4+t^2\cdot \beta}{(t+2)^2}  \ \ \ \ {\mbox    {for even} \  t, \  \ \ }  \beta \in [\frac 1{t+1},\frac 1t)\\
1+\frac{4+(t^2-1)\cdot \beta}{(t+1)\cdot(t+3)} \ \ \  \ {\mbox {for odd }  \  t, \  \ \ }  \beta \in [\frac1{t+1}, \frac 1t)  \\
\end{cases}} \ . $$

Note that the in special case $\beta=\frac{1}{t+1}$ we get the
same bounds as before. For all other cases we get a smaller bound since $\frac{4+t^2\cdot \beta}{(t+2)^2} < \beta$ is equivalent to $\beta>\frac{1}{t+1}$, and $\frac{4+(t^2-1)\cdot \beta}{(t+1)\cdot(t+3)} < \frac 1{t+1}$ is also equivalent to $\beta>\frac{1}{t+1}$.

For every interval of the form
$[\frac1{t+1}, \frac 1t)$ the function $R$ is continuous and
monotonically non-decreasing. For $t>1$, the functions are strictly increasing. The function is not continuous in general, as for example the image for $t=1$ is just the point $\frac32$, the image for $t=2$ is $[\frac 43, \frac {11}8)$, and the image for $t=3$ is $[\frac 54,\frac{23}{18})$.

\begin{lemma}\label{le3}
The asymptotic approximation ratio of FF for the parameter $\beta$ is at least $R_1(\beta)$.
\end{lemma}
%\subsection{Proof of Lemma \ref{le3}}
\begin{proof}
We start with the case of odd $t$. Let $\gamma=\frac{2-(t-1)\cdot
\beta}{t+3}$.  Since $\beta<\frac 1t$, we have $\gamma>0$.
By $\beta \geq \frac{1}{t+1}$, we have $\gamma \leq
\frac{2-(t-1)/(t+1)}{t+3}=\frac{t+3}{(t+1)\cdot(t+3)}=\frac{1}{t+1}\leq
\beta$.
We will use the property $\frac {t-1}2 \cdot \beta + \frac {t+3}2 \cdot \gamma=1$ in what follows.
Let $N>t^2$ be a large positive integer, Let $k=\lceil N\cdot(1-\gamma) \rceil -1$ (where $k< N \cdot (1-\gamma)$), and let $\delta=\frac1{N}$.

The input consists of $N\cdot (t+1)$ items of size $\gamma$, $N \cdot (t+1)$ items of size $\beta$, and $k\cdot N\cdot (t+1)$ items of size $\delta$. We consider a solution that has $N \cdot (t+1)$ bins. The $N\cdot (t+1)$ bins have identical contents, where each bin has $k$ items of size $\delta$, one item of size $\gamma$, and one item of size $\beta$.
Since $k\cdot \delta+\gamma= \frac{k}{N}+\gamma < \frac{N\cdot(1-\gamma)}N+\gamma =1$, every bin contains a load below $1$ if its item of size $\beta$ is excluded, and the packing is valid.

FF receives the items in the following order. First, all items of sizes $\delta$ arrive. Every bin for such items has exactly $N$ items and cannot receive any other items (since there are no smaller items). The number of bins so far is $k\cdot(t+1)$. Next, the input will contain repeated batches of the following items. Every batch has $\frac{t+1}2$ items of every size, such that it has $\frac{t+1}2$ items of size $\beta$ followed by $\frac{t+1}2$ items of size $\gamma$. There are $2N$ such batches. Since $\frac{t+1}2 \cdot \gamma+\frac{t+1}2 \cdot \beta = 1-\gamma$, every batch can be packed into one bin. Given the items of a batch, adding an item of size $\gamma$ (or of size $\beta\geq\gamma$) results in a load of $1+\beta$ (or a larger load), the total size after excluding one item of size $\beta$ is at least $1$, so the items of any batch cannot be packed into bins of other batches (and they cannot be packed into bins with items of size $\delta$). Thus, FF packs the input into $k(t+1)+2N$ bins.

Since $k\geq N\cdot(1-\gamma)-1$ by the definition of $k$, and by the value of $\gamma$,  we find that the number of bins of FF is at least $(t+1)\cdot (N\cdot(1-(2-(t-1)\cdot
\beta)/(t+3))-1)+2N$, while an optimal solution has at most $N(t+1)$ bins. The ratio is at least ${1-(2-(t-1)\cdot
\beta)/(t+3))}-\frac 1N+\frac{2}{t+1}=1+\frac{2}{t+1}-\frac{2-(t-1)\cdot\beta}{t+3}-\frac 1N=1+\frac{(t^2-1)\cdot \beta+4}{(t+1)\cdot(t+3)}-\frac 1N$. We get the lower bound by letting $N$ grow to infinity.

We continue with the case of even $t$. Let $\gamma=\frac{2-t\cdot
\beta}{t+2}>0$. By $\beta \geq \frac{1}{t+1}$, we have $\gamma \leq
\frac{2-t/(t+1)}{t+2}=\frac{t+2}{(t+1)\cdot(t+2)}=\frac{1}{t+1}\leq
\beta$.
We will use the property $\frac t2 \cdot \beta + \frac {t+2}2 \cdot \gamma=1$ in what follows.
Let $N>t^2$ be a large positive integer, Let $k=\lceil N\cdot(1-\gamma) \rceil -1$ (where $k< N \cdot (1-\gamma)$), and let $\delta=\frac1{N}$.

The input consists of $N\cdot t$ items of size $\gamma$, $N \cdot (t+2)$ items of size $\beta$, and $k\cdot N\cdot t+2N(N-1)$ items of size $\delta$. We consider a solution that has $N \cdot (t+2)$ bins. The first $N\cdot t$ bins have identical contents, where each bin has $k$ items of size $\delta$, one item of size $\gamma$, and one item of size $\beta$. There are also $2N$ bins, each with $N-1$ items of size $\delta$ and one item of size $\beta$. Since $k\cdot \delta+\gamma= \frac{k}{N}+\gamma < \frac{N\cdot(1-\gamma)}N+\gamma =1$, and $(N-1)\delta<1$, every bin contains a load below $1$ if its item of size $\beta$ is excluded, and the packing is valid. All items are packed, and this solution has $N(t+2)$ bins.

FF receives the items in the following order. First, all items of sizes $\delta$ arrive. Every bin for such items has exactly $N$ items and cannot receive any other items. The number of bins so far is $k\cdot t+2(N-1)$. Next, the input will contain repeated batches of the following items. Every batch has $t+1$ items in total, such that it has $\frac{t+2}2$ items of size $\beta$ followed by $\frac{t}2$ items of size $\gamma$. There are $2N$ such batches. Since $\frac{t}2 \cdot \gamma+\frac{t}2 \cdot \beta = 1-\gamma$, every batch can be packed into one bin. Since adding an item of size $\gamma$ (or of size $\beta\geq\gamma$) results in a load of $1$ (or larger) excluding one item of size $\beta$, the items of any batch cannot be packed into bins of other batches (and they cannot be packed into bins with items of size $\delta$). Thus, FF packs the input into $k\cdot t+2(N-1)+2N$ bins.

Since $k\geq N\cdot(1-\gamma)-1$ by the definition of $k$, and by the value of $\gamma$,  we find that the number of bins of FF is at least $t\cdot (N\cdot(1-(2-t\cdot
\beta)/(t+2))-1)+4N-2$, while an optimal solution has at most $N(t+2)$ bins. The ratio is at least $ \frac{t}{t+2} \cdot (1-\frac{2-t\cdot\beta}{t+2}) +\frac{4}{t+2}-\frac{1}{N}=\frac{4\cdot(t+2)+t(t+t\cdot\beta)}{(t+2)^2}-\frac 1N=1+\frac{t^2\cdot\beta+4}{(t+2)^2}-\frac 1N$. We get the lower bound by letting $N$ grow to infinity.
\end{proof}

%\medskip

To prove an upper bound, we define a weight function. For $x < \frac{1}{t+1}$, we let $w(x)=x$, and for $x\in [\frac{1}{t+1},\beta)$, we let $w(x)=R_1(x)-1$, i.e., $$w(x)=\frac{4+t^2\cdot x}{(t+2)^2} {\mbox{ \ \  if  }} t {\mbox{  is even, and  \ }}
w(x)=\frac{4+(t^2-1)\cdot x}{(t+1)\cdot(t+3)} {\mbox{\ \ \ if  }} t \mbox{ is odd.}   $$ The weight function is continuous and monotonically non-decreasing. The function is in fact equal to $\min\{x,R_1(x)-1\}$, since $\frac{4+t^2\cdot x}{(t+2)^2}<x$  and $\frac{4+(t^2-1)\cdot x}{(t+1)\cdot(t+3)}<x$ are both  equivalent to $x>\frac{1}{t+1}$.

%\medskip

\begin{lemma}\label{fourl}
The asymptotic approximation ratio of FF for the parameter $\beta$ is at most $R_1(\beta)$.
\end{lemma}
%\subsection{Proof of Lemma \ref{fourl}}
\begin{proof}
We show that the weight of a bin of an optimal solution does not exceed $R_1(\beta)$. For any bin, consider its items excluding the largest item. The weight of any item does not exceed its size, and therefore the total weight of these items is below $1$. Since the weight function is monotonically non-decreasing and no item has size above $\beta$, the weight of the additional item is at most $R_1(\beta)-1$. The total weight is therefore below $R_1(\beta)$.

We will show that for an output of FF with $Z$ bins, the total weight is at least $Z-3$. Recall that for every bin in an output of FF which is not the last bin, the load is at least $1$, since any bin whose load is below $1$ can receive at least one additional item. We say that an item is large if its size is at least $\frac{1}{t+1}$, and otherwise it is small. We will use the property that $w(\frac{1}{t+1})=\frac 1{t+1}$. FF has the property that removing a bin with its items would not change the output for the remaining items. Thus, we can remove bins whose total weights are at least $1$ and consider the remaining bins. We also remove the last output bin (and it is left to show that the total weight for the remaining bins is at least their number minus $2$).
Each one of the remaining bins has a load of at least $1$, but its total weight is below $1$.

Since all items have sizes below $\frac 1t$, every bin (with load $1$ or larger) has at least $t+1$ items. Moreover, the number of large items cannot exceed $t+1$, since the total size of $t+1$ large items is at least $1$.

We classify bins according the numbers of large items that they have, where these numbers are integers in $[0,t+1]$. The type $i$ of a bin is the number of its large items, and $i$ satisfies $0\leq i \leq t+1$.
We will show first that every bin of type $0$ or $t+1$ has weight of $1$ or more. For a bin of type $0$, its has no large items, but its load is at least $1$. Each of its items has a weight equal to its size, and thus the total weight is at least $1$, so there are no such bins. For a bin of type $t+1$, its has $t+1$ large items, each of weight at least $\frac{1}{t+1}$, and the total weight is again at least $1$, so there are no bins of this type either. For any $i=1,2,\ldots,t$, we will consider the bins of type $i$ independently of other types (since FF would create the same bins as before for any of these sub-inputs that consist of the items of specific subsets of bins). For every such subset of bins, we will show that its total weight is at least the number of bins minus $\frac 2{t}$, which will prove the claim regarding the total weight for the entire output.

Consider a fixed type $i$, where $1\leq i \leq t$. We will now analyze the subset of bins of this type.
Since every bin has at least $t+1$ items, every such bin also has at least $t+1-i$ small items (of sizes smaller than $\frac 1{t+1}$), and we let $m=t+1-i$. Since $i\leq t$, every bin of this type has at least one small item. Since $i\geq 1$, every bin of type $i$ has at least one large item, and thus its largest item is large.

Let $k$ the number of bins of  type $i$, let $Y_j$ be the total size of small items of the $j$th bin for $j=1,2,\ldots,k$ (where indexes are defined according to the order that the bins were opened), and let $X_j$ denote the total size of large items of the $j$th bin, excluding the largest item of this bin.
The smallest small item of bin $j$ has size of at most $\frac{Y_j}{m}$, since the bin has at least $m$ small items, whose total size is $Y_j$.
Moreover, since the weight of this bin is below $1$, and the total weight of $i$ large items is at least $\frac{i}{t+1}$, we find that the combined total weight of small items is below $\frac{t+1-i}{t+1}$. Since the weight of a small item is equal to its size, we find that $Y_j < \frac{t+1-i}{t+1}=\frac{m}{t+1}$ holds for $1 \leq j \leq k$.

We claim that for $j<k$ we have $X_j+Y_j+\frac{Y_{j+1}}{m} \geq 1$ holds. Indeed, the smallest small item of bin $j+1$ was not packed into bin $j$ because the load together with it (and excluding the largest item of that time) would have been at least $1$ (and the load cannot decrease over time). For bin $k$, the total size of all items is at least $1$, and the size of the largest item is at most $\beta < \frac 1t$, and therefore, $X_j+Y_j+\frac 1t \geq 1$ holds. For simplicity of notation, we let $Y_{k+1}= \frac {m}t$. Now we have $X_j+Y_j+\frac{Y_{j+1}}{m} \geq 1$ for $1 \leq j \leq k$.

We consider the case $i=1$ separately. In this case $X_j=0$ holds for every bin, and $Y_j+\frac{Y_{j+1}}{m} \geq 1$, and therefore $Y_j \geq \frac{t}{t+1}$ for $j<k$ (since $\frac{Y_{j+1}}m<\frac {1}{t+1}$) and $Y_k \geq \frac {t-1}t$. We get that the total weight of all small items of bin $j<k$ is $Y_j \geq \frac{t}{t+1}$ while the weight of the large item is at least $\frac{1}{t+1}$. Thus, the total weight for every bin $j<k$ of this type is at least $1$ (so there are no such bins), and the total weight of the last bin of this type is at least $\frac{t-1}{t}+\frac{1}{t+1}=\frac{t^2+t-1}{t(t+1)}=1 -\frac{1}{t(t+1)} > 1 -\frac 2t$.

We are left with the cases $i=2,3,\ldots,t$, and in particular, we have $t \geq 2$.
In the case $i \geq 2$, every bin of this type has at least one large item in addition to the largest one.
The total size of the large items of bin $j$ (including the largest one) is at least $X_j \cdot \frac{i}{i-1}$, since the average size of a large item (which is not the largest one) is $\frac{X_j}{i-1}$, and the largest large item has at least this size. Thus, the total size of large items is at least $X_j \cdot \frac{i}{i-1} $.

\medskip
%\medskip
%\medskip
%\medskip

We will use the weights for $i$ large items as they were defined, and for the small items, the total weight is at least the total size.
In the case where $t$ is odd, we find that the total weight for bin $j<k$ is at least $$Y_j+\frac{4\cdot i}{(t+1)\cdot(t+3)}+\frac{t-1}{t+3}\cdot X_j \cdot \frac{i}{i-1} \ . $$
In the case where $t$ is even, we find that the total weight for bin $j<k$ is at least $$Y_j+\frac{4\cdot i}{(t+2)^2}+\frac{t^2}{(t+2)^2}\cdot X_j \cdot \frac{i}{i-1} \ . $$

We have $Y_j +\frac{Y_{j+1}}m  < \frac{m}{t+1}+\frac{1}{t} = \frac{t+1-i}{t+1}+\frac 1t \leq 1 -\frac{2}{t+1}+\frac 1t=1+\frac{t+1-2t}{t(t+1)}< 1$, which holds by $Y_j < \frac{m}{t+1}$, $Y_{j+1} \leq \frac {m}{t}$, $m=t+1-i$, $t \geq 2$, and $2 \leq i \leq t$, and therefore $1 -Y_j -\frac{Y_{j+1}}m >0$, which we use in the proofs.

\noindent{\it An analysis for odd $t$.}

By using $X_j \geq 1 -Y_j -\frac{Y_{j+1}}m >0$, we get that the weight of bin $j$ is at least $$Y_j+\frac{4\cdot i}{(t+1)\cdot(t+3)}+\frac{t-1}{t+3} \cdot \frac{i}{i-1}\cdot (1-Y_j-\frac{Y_{j+1}}m)$$ $$=\frac{4i(i-1)+i(t^2-1)}{(t+1)\cdot(t+3)\cdot (i-1)}+Y_j\cdot(1-  \frac{(t-1)\cdot i}{(t+3)(i-1)} )-Y_{j+1}\cdot \frac{(t-1)\cdot i}{(t+3)(i-1)\cdot m} \ . $$

Note that $$\frac{4i(i-1)+i(t^2-1)}{(t+1)\cdot(t+3)\cdot (i-1)}=\frac{4i^2-5i+it^2}{(t+1)\cdot(t+3)\cdot (i-1)}=1+\frac{(t+2-2i)^2-1}{(t+1)\cdot(t+3)\cdot (i-1)}$$ holds since $$(t+1)\cdot(t+3)\cdot(i-1)+(t+2-2i)^2-1=(t^2+4t+3)\cdot(i-1)+ (t+2)^2-4i(t+2)+4i^2-1$$ $$=t^2\cdot i+4ti+3i-t^2-4t-3+t^2+4t+4-4it-8i+4i^2-1=t^2\cdot i+4i^2-5i \ . $$

Thus, after subtracting $1$ for every bin, we consider the following sum (for which we show that it is at least $-\frac 2t$): $$\sum_{j=1}^k (\frac{(t+2-2i)^2-1}{(t+1)(t+3)\cdot(i-1)}+Y_j\cdot(1-  \frac{(t-1)\cdot i}{(t+3)(i-1)} )-Y_{j+1}\cdot \frac{(t-1)\cdot i}{(t+3)(i-1)\cdot {m}})$$ $$=\sum_{j=1}^k (\frac{(t+2-2i)^2-1}{(t+1)(t+3)\cdot(i-1)}+Y_j\cdot(1-  \frac{(t-1)\cdot i}{(t+3)(i-1)} )-Y_{j}\cdot \frac{(t-1)\cdot i}{(t+3)(i-1)\cdot m})$$ $$+Y_{1}\cdot \frac{(t-1)\cdot i}{(t+3)(i-1)\cdot m}-Y_{k+1}\cdot \frac{(t-1)\cdot i}{(t+3)(i-1)\cdot m} \ . $$

Since $Y_1 \geq 0$ and $Y_{k+1}=\frac {m}t$, we have $$Y_{1}\cdot \frac{(t-1)\cdot i}{(t+3)(i-1)\cdot m}-Y_{k+1}\cdot \frac{(t-1)\cdot i}{(t+3)(i-1)\cdot m} \geq -\frac{(t-1)\cdot i}{t(t+3)(i-1)} \ . $$ By $i \leq 2(i-1)$ (which holds for $i\geq 2$) and $t-1 \leq t+3$, the last expression is at least $-\frac{2}{t}$.

It is therefore sufficient to show that $$\frac{(t+2-2i)^2-1}{(t+1)(t+3)\cdot(i-1)}+Y_j\cdot(1-  \frac{(t-1)\cdot i}{(t+3)(i-1)} )-Y_{j}\cdot \frac{(t-1)\cdot i}{(t+3)(i-1)\cdot m} \geq 0$$ holds for any $j$.
Since  $(t+1)(t+3)(i-1)$ is positive, it remains to prove that $$(t+2-2i)^2-1+Y_j\cdot(t+1)\cdot ((t+3)(i-1)-  (t-1)\cdot i - \frac{(t-1)\cdot i}m )\geq 0$$  holds for any $j$.

The multiplier of $Y_j(t+1)$ (which is non-negative) is $((t+3)(i-1)-  (t-1)\cdot i - \frac{(t-1)\cdot i}m )$, and it is not necessarily non-negative. If it is non-negative, we are done since $(t+2-2i)^2-1$ is also non- negative as $t+2$ is odd and $2i$ is even, so $(t+2-2i)^2\geq 1$.
Otherwise, in the case that the multiplier of $Y_j$ is negative, we would like to prove that $$(t+2-2i)^2-1+Y_j\cdot(t+1)\cdot ((t+3)(i-1)-  (t-1)\cdot i - \frac{(t-1)\cdot i}m )\geq 0$$ holds.
We use the fact $Y_j<\frac{m}{t+1}$, and we get that it is sufficient to prove that $(t+2-2i)^2-1+m\cdot ((t+3)(i-1)-  (t-1)\cdot i - \frac{(t-1)\cdot i}m )\geq 0$ holds. This is equivalent to
$(t+2-2i)^2-1+m\cdot (t+3)(i-1)-  (m+1)(t-1)\cdot i \geq 0$

By $m=t+1-i$, we will prove  $(t+2-2i)^2-1 -(t-1)\cdot i \cdot (t+2-i)+(t+1-i)(t+3)\cdot (i-1) \geq 0$ (in fact, we prove that it holds with equality).

Indeed we have  $$(t+2-2i)^2-1 -(t-1)\cdot i \cdot (t+2-i)+(t+1-i)(t+3)\cdot (i-1) $$ $$ =(t+2-2i)\cdot(t+2-2i-i(t-1)+(t+3)(i-1))-1-i^2(t-1)+(i-1)^2(t+3) $$ $$=(t+2-2i)\cdot(2i-1)-1-i^2t+i^2+i^2t-2it+t+3i^2-6i+3 $$ $$=
(2it+4i-4i^2-t-2+2i)+4i^2-2it+t-6i+2=0 \ . $$

\medskip
%\medskip
%\medskip
%\medskip

\noindent{\it An analysis for even $t$.}

By using $X_j \geq 1 -Y_j -\frac{Y_{j+1}}{m} >0 $, we get that the weight of bin $j$ is at least $$Y_j+\frac{4\cdot i}{(t+2)^2}+\frac{t^2}{(t+2)^2} \cdot \frac{i}{i-1}\cdot (1-Y_j-\frac{Y_{j+1}}{m})$$ $$=\frac{t^2\cdot i+4i^2-4i}{(t+2)^2\cdot (i-1)}+Y_j\cdot(1-  \frac{t^2\cdot i}{(t+2)^2\cdot (i-1)} )-Y_{j+1}\cdot \frac{t^2\cdot i}{(t+2)^2\cdot (i-1)\cdot m}  \ . $$

Note that $\frac{t^2\cdot i+4i^2-4i}{(t+2)^2\cdot(i-1)}=1+\frac{(t+2-2i)^2}{(t+2)^2\cdot(i-1)}$ since $$(t+2)^2\cdot(i-1)+(t+2-2i)^2=t^2\cdot i+4i^2-4i \ . $$

Thus, we consider the following sum: $$\sum_{j=1}^k (\frac{(t+2-2i)^2}{(t+2)^2\cdot(i-1)}+Y_j\cdot(1-  \frac{t^2\cdot i}{(t+2)^2\cdot (i-1)} )-Y_{j+1}\cdot \frac{t^2\cdot i}{(t+2)^2\cdot (i-1)\cdot {m}}) $$ $$=\sum_{j=1}^k (\frac{(t+2-2i)^2}{(t+2)^2\cdot(i-1)}+Y_j\cdot(1-  \frac{t^2\cdot i}{(t+2)^2(i-1)} )-Y_{j}\cdot \frac{t^2\cdot i}{(t+2)^2(i-1)\cdot m}) $$ $$+Y_{1}\cdot \frac{t^2\cdot i}{(t+2)^2(i-1)\cdot m}-Y_{k+1}\cdot \frac{t^2\cdot i}{(t+2)^2(i-1)\cdot m} \ . $$

Since $Y_1 \geq 0$ and $Y_{k+1}=\frac {m}t$, we have $$Y_{1}\cdot \frac{t^2\cdot i}{(t+2)^2(i-1)\cdot m}-Y_{k+1}\cdot \frac{t^2\cdot i}{(t+2)^2(i-1)\cdot m} \geq -\frac{t\cdot i}{(t+2)^2(i-1)} \ . $$ By $i \leq 2(i-1)$ (which holds for $i\geq 2$) and $t \leq t+2$, the last expression is at least $-\frac{2}{t+2} \geq -\frac 2t $.

It is therefore sufficient to show that $\frac{(t+2-2i)^2}{(t+2)^2\cdot(i-1)}+Y_j\cdot(1-  \frac{t^2\cdot i}{(t+2)^2(i-1)} )-Y_{j}\cdot \frac{t^2\cdot i}{(t+2)^2(i-1)\cdot m} \geq 0$ holds for any $j$. The multiplier of $Y_j$ is $(1- \frac{t^2\cdot i}{(t+2)^2(i-1)}-\frac{t^2\cdot i}{(t+2)^2(i-1)\cdot m})$. If it is non-negative, we are done since $\frac{(t+2-2i)^2}{(t+2)^2\cdot(i-1)}$ is non-negative (and so is $Y_j$). Otherwise, since the denominator is positive, we would like to prove that $(t+2-2i)^2-Y_j(t^2\cdot i \cdot (1+\frac 1{m})-(t+2)^2\cdot (i-1)) \geq 0$ holds.
We use the fact $Y_j<\frac{m}{t+1}$, and we get $(t+2-2i)^2-Y_j(t^2\cdot i \cdot (1+\frac 1{m})-(t+2)^2\cdot (i-1)) \geq (t+2-2i)^2-\frac{1}{t+1}(t^2\cdot i \cdot (m+\frac{m}{m})-(t+2)^2\cdot (i-1)\cdot m)$. We get that it is sufficient to prove $$(t+2-2i)^2-\frac{1}{t+1}(t^2\cdot i \cdot (m+1)-(t+2)^2\cdot (i-1)\cdot m) \geq 0 \ . $$ By $m=t+1-i$, we will prove  $(t+1)\cdot(t+2-2i)^2-t^2\cdot i \cdot (t+2-i)+(t+2)^2\cdot (i-1)\cdot (t+1-i) \geq 0$ (in fact, we prove that it holds with equality).
Indeed we have
%$(t+1)\cdot(t+2-2i)^2+(t+2)^2\cdot (i-1)\cdot (t+1-i)=(t+2)^2\cdot ((t+1)+(i-1)(t+1)-i(i-1))+(t+1)(4i^2-4i(t+2))=(t+2)^2(i\cdot t+2i-i^2)+4i^2(t+2)-4i(t+2)^2+4i(t+2)-4i^2=i((t+2)^2( t+2-i)+4i(t+2)-4(t+2)^2+4(t+2)-4i)=i((t+2) %((t+2)((t-2-i)+4i)+4(t+2)-4i)=i((t+2) (t^2-4-i\cdot t+2i)+4(t+2-i))=i(t+2)(t+2-i)(t-2)+4i(t+2-i)=i(t+2-i)(t^2-4+4)=it^2(t+2-i)$.
$(t+1)\cdot(t+2-2i)^2+(t+2)^2\cdot (i-1)\cdot (t+1-i)=it^2(t+2-i)$.
\end{proof}

We conclude with the following theorem. \begin{theorem}
The asymptotic approximation ratio of FF and the parameter $\beta$ is $R_1(\beta)$ if $\beta<1$ and for $\beta=1$ the ratio is $2$.
\end{theorem}

We found that FF has a smaller asymptotic approximation ratio than that of NF and WF unless $\beta$ is a reciprocal of an integer. In the next section we will see that sorting sometimes improves the performance but the improvement is not large.

\subsection{Algorithms with sorting}\label{wsort}

We consider item sizes in $(0,\beta]$ again, where $0<\beta\leq 1$.
The algorithm Next Fit Decreasing (NFD) uses a sorted list on the items (sorted by size in a non-increasing order). Every bin will receive a maximum length prefix of the unpacked list of items that can be packed into a bin. That is, a maximum length prefix of total size strictly below $1$ together with one additional item. We show that the asymptotic approximation ratio of this algorithm not better than those without sorting, but it is better for some values of $\beta$. For example, for $\beta=\frac 23$, we show that NFD has an asymptotic approximation ratio of $\frac 32$, while NF has ratio $\frac 53$. For $\beta=0.4$, the ratio for NFD is $\frac 43$, it is $1.4$ for NF, and it is  $1.35$ for FF.

\begin{proposition}\label{pro6}
The asymptotic approximation ratio of NFD for max-OEBP and the parameter $\beta$ is  exactly $1+\frac{1}{t+1}$.
\end{proposition}
%\subsection{Proof of Proposition \ref{pro6}}
\begin{proof}
For the lower bound, let $N>0$ be a large integer. The input consists of items of size $\frac 1{t+1}$ (large items) and items of size $\frac{1}{N\cdot (t+1)}$ (small items). The number of large items is $N\cdot (t+1)$, and the number of small items is $N(t+1)(N(t+1)-1)$. An offline solution has $N\cdot (t+1)$ bins, where every bin has $N\cdot (t+1)-1$ small items and one large item. Those bins are valid since the total size of $N(t+1)-1$ small items  is $(N(t+1)-1)\cdot \frac{1}{N\cdot (t+1)}=1-\frac{1}{N(t+1)}$.

On the other hand, NFD packs the items as follows. First, it packs $N$ bins with $t+1$ large items each, where every bin has a total size of items of $1$. The last such bin can also receive some small items. Specifically, it receives $N-1$ small items, since the total size of $N-1$ small items and $t$ large items is $(N-1)\cdot \frac{1}{N(t+1)}+t\cdot \frac {1}{t+1}=\frac 1{t+1}-\frac 1{N(t+1)}+1-\frac 1{t+1} =1-\frac{1}{N(t+1)}$, while an additional small item brings the total size (of all items excluding one large item) to $1$. The remaining small items are packed into new bins, where every bin except for possibly the last one, will have $N(t+1)$ items. Since $\lceil \frac{N(t+1)(N(t+1)-1)-(N-1)}{N(t+1)} \rceil=N(t+1)-1$, NFD uses $N+(t+1)\cdot N-1$ bins. Recall that an optimal solution has $N(t+1)$ bins. Thus, the  approximation ratio for a specific value of $N$ is $1+\frac 1{t+1}-\frac 1{N(t+1)}$ and the asymptotic approximation ratio is at least $1+\frac 1{t+1}$.

For an upper bound, consider the weight function $w$ where $w(x)=x$ for $x<\frac 1{t+1}$ and $w(x)=\frac 1{t+1}$ for $\frac 1{t+1} \leq x \leq \beta$. Note that in the cases $\beta=\frac 1{t+1}$, the weight function is simply the identity function. In other cases, it always holds that $w(x) \leq x$, and $w$ is monotonically non-decreasing. Consider a fixed bin $B$ of an optimal solution.
The weight of the largest item is at most $\frac 1{t+1}$ (since this is the maximum possible weight), and the total size of remaining items is below $1$, so their total weight is also below $1$. In total, the total weight for $B$ is below $1+\frac 1{t+1}$.

Consider an input $I$ and the solution of NFD. Let $Q$ denote the number of items of sizes in $[\frac 1{t+1},\beta]$. The number $NFD(I)$ consists of $\lceil \frac{Q}{t+1} \rceil$ bins containing items of sizes at least $\frac 1{t+1}$, where excluding the last such bin, every bin has exactly $t+1$ such items, and its weight is exactly $1$. The last such bin may have additional (smaller) items, and its total weight may be below $1$. Bins opened afterwards will only have items whose sizes are below $\frac 1{t+1}$ and all these bins have total sizes of at least $1$ (and thus total weights of at least $1$), maybe expect for the last bin used to pack the input.  We find that the total weight of all bins is at least $NFD(I)-2$.
\end{proof}

Next, we show that in the general case the ratio is smaller for FF. FFD was analyzed in the past for items of sizes in $(0,1)$ \cite{theFFD}. While the upper bound is not valid if there may be items of size $1$ in the input, the lower bound construction is very similar to that of \cite{theFFD} and it is included for completeness.

\begin{proposition}\label{pro7}
The asymptotic approximation ratio of FFD for max-OEBP is $\frac 32$.
\end{proposition}
%\subsection{Proof of Proposition \ref{pro7}}
\begin{proof}
For the lower bound, let $N>10$ be a large integer. The input consists of items of size $1-\frac 1{N}$ (large items) and items of size $\frac{1}{N}$ (small items). The number of large items is $2N$, and the number of small items is $2N(N-1)$. An offline solution has $2N$ bins, where every bin has $N-1$ small items and one large item. Those bins are valid since the total size of $N-1$ small items  is $1- \frac{1}{N}$.

On the other hand, FFD packs the items as follows. First, it packs $N$ bins with two large items each, where every bin has a total size of items of $2-\frac{2}N$. The total size for each such bin, excluding an item of maximum size, is $1-\frac 1N$, and adding one small item will bring this total size to $1$. Thus, the bin will not receive additional items. The small items are packed into new bins, where every bin will have $N$ items, and the number of bins is $2(N-1)$. The total number of bins for FFD is $3N-2$.
Thus, the  approximation ratio for a specific value of $N$ is $\frac{3N-2}{2N} = 1.5-\frac 1N$ and the asymptotic approximation ratio is at least $1.5$.

For an upper bound, consider the weight function $w$ where $w(x)=x$ for $x<\frac 12$ and $w(x)=\frac 12$ for $\frac 12 \leq x \leq 1$. It always holds that $w(x) \leq x$, and $w$ is monotonically non-decreasing. Consider a fixed bin $B$ of an optimal solution.
The weight of the largest item is at most $\frac 12$ (since this is the maximum possible weight), and the total size of remaining items is below $1$, so their total weight is also below $1$. In total, the total weight for $B$ is below $1+\frac 12=\frac 32$.

Consider an input $I$ and the solution of FFD. First, items of size $1$ are packed into separate bins. These bins (if they exist) may receive additional items later.
Every solution packs each item of size $1$ into a distinct bin, and therefore the case where no other bins are used by FFD yields an optimal solution. We consider the case where FFD uses at least one bin that has no item of size $1$. Let $Q$ denote the number of items of sizes in $[\frac 12,1)$, and let $D$ denote the number of items of size $1$.  The items of sizes in $[\frac 12,1)$ are packed right after the items of size $1$. Every bin with an item of size $1$ can receive exactly one item of size in $[\frac 12,1)$. Thus, in the case $Q\leq D$, every item of size in $[\frac 12,1)$ is packed with an item of size $1$, and we consider this case later. Otherwise, there are $Q-D>0$ remaining items of sizes in $[\frac 12,1)$, and they are packed into new bins, such that every new bin has two such items (and weight at least $1$), except for possibly one bin. Thus, there are $\lceil \frac{Q+D}2 \rceil$ bins with items of sizes of at least $\frac 12$, where the last bin may have only one such item while the other bins have two such items (all these bins may also have other items). These bins have total weights of at least $1$, possibly except for the last such bin. The other bins have items of sizes below $\frac 12$, and every such bin has total size of at least $1$ (and therefore a total weight of at least $1$), possibly expect for the last bin packed by FFD. We find that the total weight of all bins is at least $FFD(I)-2$.

We are left with the case that not all bins with items of size $1$ received additional items before items of sizes below $\frac 12$ are packed, and the number of bins is at least $Q+1$. There are three types of bins.
Bins with two items of sizes at least $\frac 12$ and weight $\frac 12$ each (and possibly other items), where the number of such bins is $D\geq 0$. Bins with an item of size $1$ and items of sizes below $\frac 12$, where the number of such bins is $Q-D>0$. Finally, there are $FFD(I)-D>0$ bins with only items of sizes below $\frac 12$. For the last type of items, we already explained that the total weights are at least $1$, except for possibly the last bin. It is left to consider the remaining $D-Q$ bins.

We consider the $D-Q$ bins with one item of size $1$ and possibly items of sizes below $\frac 12$. If the total size of items of sizes below $\frac 12$ is at least $\frac 12$, the total weight for the bin is at least $1$. If all these bins satisfy this last property, we get  that the total weight of all bins is at least $FFD(I)-1$. This is indeed the case because there is at least one bin without an item of size $1$, and it has at least one item of size below $\frac 12$. This item could not be packed into a bin with an item of size $1$, and therefore its items of sizes below $\frac 12$ (all its items excluding the item of size $1$) already had a total size above $\frac 12$, and this total size cannot decrease later.
\end{proof}

We analyze integer values of $t\geq 1$ such that all items have sizes in $(0,\frac 1t)$. The case $t=1$ in fact also follows from the general case, which is obvious for the upper bound, and holds for the lower bound since the construction of the input consists of items of sizes below $1$. In fact, in the case $t=1$, this is simply the result of \cite{theFFD}.
Note that we usually assume in this paper that the input may have items of size $\beta$, but here we assume that it cannot be obtained and therefore we do not use the parameter $\beta$. We analyze this case due to its simplicity, and since it is similar to the previous proof and the proof of \cite{theFFD}.

\begin{proposition}\label{pro8}
The asymptotic approximation ratio of FFD and items of sizes in $(0,\frac 1t) $ for max-OEBP is $1+\frac 1{t+1}$.
\end{proposition}
%\subsection{Proof of Proposition \ref{pro8}}
\begin{proof}
For the lower bound, let $N>0$ be a large integer. The input consists of items of size $\frac 1t-\frac 1{t\cdot N}$ (large items) and items of size $\frac{1}{N}$ (small items). The number of large items is $(t+1)\cdot N$, and the number of small items is $N(N-1)\cdot(t+1)$. An offline solution has $N(t+1)$ bins, where every bin has $N-1$ small items and one large item. Those bins are valid since the total size of $N-1$ small items  is $1- \frac{1}{N}$.

On the other hand, FFD packs the items as follows. First, it packs $N$ bins with $t+1$ large items each, where every bin has a total size of items of $\frac{t+1}t-\frac{t+1}{tN}$. The total size for each such bin, excluding an item of maximum size, is $1-\frac 1N$, and adding one small item will bring this total size to $1$. Thus, the bin will not receive additional items. The small items are packed into new bins, where every bin will have $N$ items, and the number of bins is $(N-1)\cdot(t+1)$. The total number of bins for FFD is $N+(N-1)(t+1)$.
Thus, the  approximation ratio for a specific value of $N$ is $\frac{Nt+2N-t-1}{(t+1)N} = \frac{t+2}{t+1}-\frac 1N$ and the asymptotic approximation ratio is at least $1+\frac{1}{t+1}$.

For an upper bound, consider the weight function $w$ where $w(x)=x$ for $x<\frac 1{t+1}$ and $w(x)=\frac 1{t+1}$ for $\frac 1{t+1} \leq x < \frac 1t$. It always holds that $w(x) \leq x$, and $w$ is monotonically non-decreasing. Consider a fixed bin $B$ of an optimal solution.
The weight of the largest item is at most $\frac 1{t+1}$ (since this is the maximum possible weight), and the total size of remaining items is below $1$, so their total weight is also below $1$. In total, the total weight for $B$ is below $1+\frac 1{t+1}$.

Consider an input $I$ and the solution of FFD. Items of sizes in $[\frac 1{t+1},\frac 1t)$ are packed first, such that every bin has exactly $t+1$ items, possibly except for the last bin with such items. It is possible that such bins will get other items later, but not items of sizes in this interval. Afterwards, other items are packed, and every bin (except for possibly the last one) has a total size of items of at least $1$. Bins with $t+1$ items of sizes at least $\frac{1}{t+1}$ have totals weights of at least $1$ since each such item has weight $\frac {1}{t+1}$. Bins with only items of sizes below $\frac 1{t+1}$ have total weights not smaller than total sizes of items. Thus, there are at most two bins with total weights below $1$.
\end{proof}

In order to show why the other parametric cases are harder, including the cases of the form $(0,\frac 1t)$, we analyze one additional case where item sizes are in $(0,\frac 12]$.
\begin{proposition}\label{pro9}
The asymptotic approximation ratio of FFD for $\beta=\frac 12$ for max-OEBP is $\frac 43$.
\end{proposition}

%\subsection{Proof of Proposition \ref{pro9}}
\noindent\begin{proof}
The lower bound for this case was already proved in the previous proposition, where in particular we studied the case with items in $(0,\frac 12)$. We show an upper bound of the same value, which requires taking care of items of size exactly $\frac 12$. Thus, we analyze the packing more carefully.

First, we extend the weight function $w$, and let $w(x)=x$ for $x<\frac 1{3}$ and $w(x)=\frac 1{3}$ for $\frac 1{3} \leq x \leq \frac 12$.
Any bin of an offline solution has weight of at most $\frac 43$ is once again due to the property that no item has weight above $\frac 13$, and no item has weight larger than its size. When we use this function, we will only need to consider the bins of FFD. This function will assist us in some cases of the proof but not all of them.
 For example, given a large positive integer $N$, consider an input with $2N$ items of size $\frac 12$ followed by $N$ items of size $0.28$. In an FFD packing, every bin gets two items of size $\frac 12$ and one item of size $\frac 7{24}\approx 0.29167$. These bins have weights of $\frac{23}{24}<1$, and there are many such bins ($N$ bin), so this is not just a matter of an additive constant. If the input stops here, the weight is too small. Even if it has other items, all of size $0.25$ (such that these items cannot be packed into the already existing bins, and every new bin has exactly four items), the total weight is not sufficient. For this input, the bins of optimal solutions will have smaller total weights than $\frac 43$, so it is not a counter-example. However, this kind of inputs cannot be taken into account easily using the simple weights, and instead we need another definition of weights that will assist us in analyzing this case. We will use a different weight function for a class of cases similar to this one. Since we aim at proving an upper bound of $\frac 43$, it does not seem possible to use a weight function where the weight of $\frac 12$ is above $\frac 13$.

If there are no items of size exactly $\frac 12$, we are done by the previous proposition. Thus, we assume that there is at least one such item.
Items of size $\frac 12$ are packed first by FFD, and they are packed in pairs, such that there is a prefix of bins with two such items, and possibly one additional bin with one such item. These bins may receive other items later. Let $Q \geq 0$ denote the number of bins with two items of size $\frac 12$. Afterwards, items of sizes in $[\frac 13,\frac 12)$ are packed. If there are at most $Q$ such items, a prefix of the bins receives one such item each. If there are at least $Q+1$ such items, every bin out of the first $Q$ bins has one such item, and the remaining such items are packed in triples (if there is a bin with one item of size $\frac 12$, it is considered next for packing, and can receive two additional items of sizes in $[\frac 13,\frac 12)$ if such items exist). In the case that there are at least $Q+1$ items of sizes in $[\frac 13,\frac 12)$, there is a prefix of bins with three items of sizes in $[\frac 13,\frac 12]$ (and possibly other items packed later), after this prefix one bin may contain one or two such items (and possibly other items), and the remaining bins have items of sizes below $\frac 13$.

Let $G$ denote the largest gap of the bins with two items of size $\frac 12$, which is defined as follows, at termination (after all items have been packed). If $Q=0$ we let $G=0$. Otherwise, for every bin $B$ with two items, let its gap $g$ be $1$ minus the total size of its items excluding its largest item (which has size $\frac 12$). The value $g$ is such that any set of items of size strictly below $g$ can be added to the bin, but it is impossible to add items of total size of $g$ or more. Note that the largest item will be still of size $\frac 12$ even if other items were added to the bin (since they are smaller). For the bins that do not have  two items of size $\frac 12$  (there are $FFD(I)-Q$ such bins for input $I$) will only have items of sizes $G$ or more, since the gaps cannot become smaller over time (an item of size no larger than $G$ can fit into a gap, and should have been packed there by the action of FFD).

Case 1. Every bin out of the first $Q$ bins has an item whose size is in $[\frac 13,\frac 12)$. We use the weight function $w(x)$.
In this case there is a prefix of bins with three items of weight $\frac 13$ packed into each bin. If there is a bin with one or two such items, it may have a smaller weight. The remaining bins, possibly except for the last bin, have items of sizes below $\frac 13$, whose weights are equal to their sizes. Thus, since the loads of these bins are at least $1$, so are their weights. The total weight for input $I$ is at least $FFD(I)-2$.

The complement of case 1 is the case where the number of items of sizes in $[\frac 13,\frac 12)$ is below $Q$, and only bins with two items of size $\frac 12$ have such items (one item per bin).

Case 2. In this case we assume that $G\leq \frac 16$, and we use the weight function $w(x)$.  For the algorithm, we show that any bin except for at most two bins has weight of at least $1$. These two bins are a bin with one item of size $\frac 12$ if it exists, and the last bin. Every bin out of the first $Q$ bins has two items of size $\frac 12$ and total weight $\frac 23$. If such a bin also has another item of size at least $\frac 13$, we are done. Otherwise, since its gap is at most $\frac 16$, the total size of its items excluding the items of size $\frac 12$ is at least $\frac 13$. These items have weights equal to their sizes, and their total weight is at least $\frac 13$. The bins without items of size $\frac 12$ that are not the last bin have total sizes of at least $1$ and the weight of every item is equal to its size, since their items have sizes below $\frac 13$.

Case 3. The last bin contains an item of size $\frac 12$. In this case the number of bins of FFD is at most $Q+1$. Letting the number of items of size $\frac 12$ be $X$, since no bin can have more than two such items, the optimal cost is at least $\lceil \frac{X+1}2 \rceil$, which is exactly the number of bins of FFD.

Case 4. In the remaining case, where the last bin has no item of size $\frac 12$ and $G>\frac 16$, we modify the input by removing all items of the last bin except for the first item ever packed into this bin. The optimal cost may only decrease, the cost of FFD is unchanged, and $G$ cannot decrease. This first item of the last bin is the smallest input item, and we denote its size by $\theta$. We have $\theta \geq G$, since this item was not combined into an earlier bin. Thus, $\theta > \frac 16$. We define a weight function $w_1$ only for sizes in $[\frac 16,\frac 12]$ as there are no other item sizes. In fact, we consider two cases, and define two different weight functions. If $\theta\geq \frac 14$, the input consists of items of sizes in $[\frac 14,\frac 12]$, and we define $w_1(x)=\frac 13$ for every item. For an optimal solution, no bin has more than four items, and therefore the total weight for each bin is at most $\frac 43$. For FFD, every bin except for possibly the last one will have at least three items each, and the total weight of every bin with at least three items is at least $1$. This holds due to the following. After the items of size $\frac 12$ are packed, the remaining items have sizes in $[\frac 14,\frac 12)$. When items smaller than $\frac 12$ are being packed, as long as a bin with two items of size $\frac 12$ does not have a third item, no new bins are created. Thus, every such bin receives one new item before new bins are created. If there is a bin with one item of size $\frac 12$, this bin will receive at least two additional items before any bin is created. Similarly, every new bin will receive at least three new items before an empty bin is used.

We are left with the case where $\frac 16 < \theta < \frac 14$. We have that $  \frac 14< \frac 12-\theta < \frac 13$.

In this case we use the following weight function. Let $$w_2(x)=\begin{cases} \vspace{0.3cm}
\frac 13  \ \ \ \ { \ for \ \ }  x  \in [\frac 12 - \theta,\frac 12]\\
\frac 14  \ \ \ \ { \ for \ \ }  x  \in [\frac 14,\frac 12-\theta)\\
\frac 15  \ \ \ \ { \ for \ \ }  x  \in [\frac 15,\frac 14)\\
\frac 16  \ \ \ \ { \ for \ \ }  x  \in [\frac 16,\frac 15)\\
\end{cases}$$ be the weight function for this case.

Consider the algorithm. Every bin with two items of size $\frac 12$ has (in addition to these two items) items of size at least $\frac 12-G \geq \frac 12-\theta$. The items of size $\frac 12$ have total weight of $\frac 23$. The bin also has an item of size at least $\frac 12-\theta$ or at least two other items. Thus, the additional weight excluding the two items of size $\frac 12$ is at least $\frac 13$, as this is the weight of an item of size at least $\frac 12-\theta$ and the weight of any two items is at least $\frac 13$.
Except for at most four bins, every remaining bin has four items of sizes in $[\frac 14,\frac 13)$ or  five items of sizes in $[\frac 15,\frac 14)$ or six items of sizes in $[\frac 16,\frac 15)$. Every bin out of such bins has weight not below $1$.

Consider a bin of an optimal solution. If no item, except for possibly the largest one, has size of at least $\frac 12-\theta$, the total weight is at most $\frac 43$, since for any item of size in $[\frac 16,\frac 12-\theta)$, the weight of the item does not exceed its size, and the largest possible weight is $\frac 13$. If the bin has at most three items (excluding the largest item), then the total weight is also at most $\frac 43$. Since all items have sizes of at least $\frac 16$, the number of items excluding the largest item is at most five. Thus, we consider the total weight of four or five items whose total size is strictly below $1$ and show that their total weight is at most $1$. By the previous discussion, we can assume that at least one of these items has size at least $\frac 12-\theta$. If there are four other items, the total size of the five items is at least
$(\frac 12-\theta)+4\theta =\frac 12+3\theta > 1$, since $\theta>\frac 16$. Thus, there are four items in total. If two of them have sizes of at least $\frac 12-\theta$, the total size of the four items is at least $2(\frac 12-\theta)+2\cdot \theta=1$, so this is impossible. We find that in the case left to consider there are four items, one of which has size at least $\frac 12-\theta$. The largest item and the second largest item have total weight of $\frac 23$, and we consider the remaining three items, whose total size is below $\frac 12 +\theta$. Since each of these items has size of at least $\theta$, the two largest items together have total size below $\frac 12$, and there is at most one item with size at least $\frac 14$. Thus, the largest item out of the three has weight at most $\frac 14$, and each of the other two items has weight of at most $\frac 15$. The three items together have a total weight of at most $0.65$, and the bin has total weight strictly below $\frac 43$.
\end{proof}

\section{Bin packing with item types for Max-OEBP}\label{ittypes}
We start with the analysis of batched bin packing, as this analysis is simple, and afterwards to proceed to the price of clustering (PoC).  We discuss algorithms for batched bin packing that find an optimal solution for every batch separately.
Once again we use a parameter $0<\beta\leq 1$, and let $t=\lceil \frac 1{\beta}\rceil-1$.
Thus, $\beta \in [\frac 1{t+1},\frac 1{t})$ if $\beta<1$ and $t=0$ if $\beta=1$.  The notation is the same as for clusters.

\begin{theorem}
For any number of batches $\ell\geq 2$, the asymptotic approximation ratio for batched bin packing and Max-OEBP is $1+\frac 1{t+1}$ and it is equal to $2$ for the general case.
\end{theorem}
\begin{proof}
We will show an upper bound for any constant $\ell \geq 2$ and a lower bound for $\ell=2$ (if $\ell$ is larger, further batches are empty). In the upper bound, the additive constant depends on $\ell$, which cannot be avoided since for an input with $\ell$ items of size $\frac 1{\ell}$ each, an optimal solution has a single bin, and in the case that every item belongs to a different batch, $\ell$ bins are needed (this is similar to standard bin packing \cite{Epstein16}).

For the lower bound, let $N$ be a large positive integers. There are $N^2\cdot (t+1)$  items of size $\frac 1{t+1}$ and $N^2\cdot (N-1)\cdot (t+1)$ items of size $\frac 1N$. An optimal solution for the entire input has identical bins with $N-1$ items of size $\frac 1N$ and one item of size $\frac 1{t+1}$, and the number of such bins is $N^2\cdot (t+1)$.
The two batches are such that one batch has all items of one size and the other one has all items of the other size. An optimal solution that packs each batch separately will have loads of $1$ for every bin, and the total number of bins is $ N^2+N(N-1)\cdot (t+1)$. Letting $N$ grow without bound, the lower bound on the asymptotic approximation ratio is $\frac 1{t+1}+1$.

For the upper bound, consider the case $\beta=1$ first. By applying FF to every batch we get a solution with separate bins for the different batches, such that an optimal solution for each batch cannot have a larger number of bins, and such that all bins have loads not smaller than $1$ except for possibly the last bin of every batch. Letting $S$ denote the total size of items, we have $S > \sum_{i=1}^{\ell} (FF(I_i)-1) \geq \sum_{i=1}^{\ell} (OPT(I_i)-1)  $. Thus, $\sum_{i=1}^{\ell} OPT(I_i) \leq S +\ell$. In an optimal solution, no bin has load of $2$ (or a larger load), so $S< 2\cdot OPT(I)$. We get $\sum_{i=1}^{\ell} OPT(I_i) < 2\cdot OPT(I) +\ell$.

Next, consider the case $t\geq 1$. We use the following weight function: $w(x)= x$ for $x \leq \frac 1{t+1}$ and $w(x)= \frac 1{t+1}$ otherwise (for $x\in  [\frac 1{t+1},\beta]$). We have $w(x) \leq x$ and $w(x)\leq \frac 1{t+1}$. Thus, for any bin of an optimal solution, the total weight is at most $1+\frac 1{t+1}$. To show  $\sum_{i=1}^{\ell} OPT(I_i) < (1+\frac 1{t+1})\cdot OPT(I) +2\ell$, we will consider an application of NFD on the items of a batch, and we show that all bins expect for at most two bins, have total weights of at least $1$ per bin. When NFD is applied, items of sizes in $[\frac 1{t+1},\frac 1{t})$ are packed first, and every bin receives $t+1$ such items. Then, a bin may receive such items and smaller items. Finally, some bins receive smaller items. Thus, bins with $t+1$ items of sizes at least $\frac 1 {t+1}$ have total weights of at least $1$. If there is a bin where some of the items are such, its weight may be smaller than $1$, but afterwards every bin except for possibly the last one has load of at least $1$, and for every item the weight is equal to the size.
\end{proof}

\subsection{The price of clustering}\label{pocc}
Finally, we discuss the PoC. Recall that the PoC is defined as a certain approximation ratio, and for the PoC we analyze the asymptotic approximation ratio and the absolute approximation ratio simultaneously.

We start with the general case and later consider the parametric case. Recall that every cluster has an optimal cost of at least $2$.

\begin{lemma}
For any input $I$ for max-OEBP, it holds that $\sum_{j=1}^{\ell} OPT_j \leq 3\cdot OPT$.
\end{lemma}

\noindent\begin{proof}
We use a simple weight function $w(x):[0,1)\rightarrow (0,0.5]$, where $w(x)=x$ for $x\in (0,\frac 12]$, and $w(x)=\frac 12$ for $x\in (0.5,1]$.

Consider a bin of an optimal solution. The total size of all items excluding the largest item is below $1$, and no item has weight above $\frac 12$. Thus, the total weight of all items excluding the largest one is below $1$, and together with the largest item, the total weight is below $1.5$, and we have $W \leq 1.5\cdot OPT$.

Consider a cluster for $I_j$, and an solution of FF for it, which contains $A_j\geq 2$ bins.  We will show that every bin that is not one of the last two bins has weight of $\frac 12$ or more, and the two last bins have a total weight of at least $1$ together.

Every bin except for possibly the last one has a total size of items not smaller than $1$ already when the next bin receives its first item, since this last item could not be packed into the previous bin.
If such a bin has an item of size above $\frac 12$, then already the weight of this item is $\frac 12$, and since weights are positive, the total weight of the items of this bin is at least $\frac 12$. Otherwise, the weight of an item is equal to its size and we get a total of $1$ or more for every bin without an item of size above $\frac 12$.

Consider the last two bins packed by $FF$ for $I_j$. If the penultimate bin has no item of size above $\frac 12$, we are done even without considering the last bin. If the two bins together have at least two such items, their total weight is at least $1$. We are left with the case where the penultimate bin has exactly one such item, and this is its largest item. The first item of the last bin was packed into the new bin since adding it to the penultimate bin would have resulted in a total size above $1$, excluding the largest item. Thus, the total weight for the two bins is in fact at least $1.5$.

Thus, we find $W_j \geq \frac{A_j}2$, and by $OPT_j \leq A_j$, we get $W=\sum_{j=1}^{\ell} W_j \geq \frac 12 \cdot \sum_{j=1}^{\ell} A_j \geq \frac 12 \cdot \sum_{j=1}^{\ell} OPT_j$.
Combining the two bounds on $W$ gives  $\sum_{j=1}^{\ell} OPT_j \leq 2\cdot W \leq 3\cdot OPT$.
\end{proof}

\begin{lemma}\label{le12}
The PoC of max-OEBP is at least $3$.
\end{lemma}
%\subsection{Proof of Lemma \ref{le12}}
\begin{proof}
We introduce a sequence of inputs for which the PoC grows to $3$ as the index $N$ grows to infinity.
For a positive integer $N>2$, there are $N(N+1)$ items whose sizes are equal to $1$, and $N(N^2-1)$ items whose sizes are equal to $\frac{1}{N}$.
An optimal solution has $N(N+1)$ bins, each with one item of size $1$ and $N-1$ items with sizes of $\frac 1N$.

The input is split into the following $N+2$ clusters. One cluster has all items of size $1$, and the only possible packing for them consists of $N(N+1)$ bins. The other items are split into $N(N-1)$ clusters with $N+1$ items each. Since the total size of any such $N$ items is $1$, the optimal cost for every such cluster is $2$, and the total cost for all clusters is therefore $N(N+1)+2N(N-1)=3N^2-N$.
We find that the PoC for the instance is $\frac{3N^2-N}{N(N+1)}=3-\frac 4{N+1}$.
\end{proof}

Next, we discuss the parametric case for this problem. Let the upper bound on the sizes be $0<\beta<1$, that is,  item sizes are in $(0,\beta]$. Let $t$ be such that $\beta \in [\frac 1{t+1},\frac 1t)$, i.e., $t=\lceil \frac{1}{\beta} \rceil-1$.

\medskip

Let $R_2(1)=3$, $R_2(0)=2$, and for $\beta<1$, $\displaystyle{R_2(\beta)=\begin{cases}
\vspace{0.3cm}
 2\cdot\frac{t+3}{t+2}  {\it{\mbox \it \ \ \ \ \ \  for \ \ \ \ \ \ }}  \beta \in [\frac 1{t+1},\frac{t+1}{t(t+2)})\\
2+\frac{2\cdot t\cdot \beta}{t+1} {\it {\mbox \it \ \ \ \ \ \ for \ \ \ \ \ \ }}  \beta \in [\frac{t+1}{t(t+2)}, \frac 1t)  \\
\end{cases}}$.

\medskip

The function $R$ is continuous for $\beta$ tending to $0$ and $1$ due to the values of one-sided limits. For values of $\beta$ of the form $\beta=\frac{t+1}{t(t+2)}$, we have $R_2(\beta)=2+\frac{2\cdot t\cdot (\frac{t+1}{t(t+2)})}{t+1}=2+\frac{2}{t+2}=\frac{2(t+3)}{t+2}$, so the function is continuous for these points. For values of $\beta$ of the form $\beta=\frac{1}{q}$ for an integer $q\geq 2$, we have $R_2(\beta)=\frac{2(q+2)}{q+1}$, and the left-hand limit is $2+\frac{2}{q+1}=\frac{2(q+2)}{q+1}$. The function is piecewise linear (constant for some parts), and monotonically non-decreasing. In particular, the image of $[0,1]$ is $[2,3]$.

\begin{lemma}
For any input $I$ for max-OEBP for which all item sizes are in $(0,\beta]$, it holds that $\sum_{j=1}^{\ell} OPT_j \leq R_2(\beta)\cdot OPT$.
\end{lemma}
\begin{proof}
We define the following weight function, where $t=\lceil \frac{1}{\beta} \rceil-1$.

$$\displaystyle{w(x)=\begin{cases}
%\vspace{0.3cm}
x {\it \mbox  {\it  \ \ \ \ \ \ for \ \ \ \ \ \ } }  x \in [0,\frac 1{t+2}) \\%\vspace{0.3cm}
\frac{1}{t+2}  {\it \mbox {\it \ \ \ \ \ \  for \ \ \ \ \ \ }}  x \in [\frac 1{t+2},\frac{t+1}{t(t+2)})\\
x\cdot \frac t{t+1} {\it {\mbox \it  \ \ \ \ \ \ for \ \ \ \ \ \ }}  x \in [\frac{t+1}{t(t+2)}, \frac 1t)  \\
\end{cases}} \ . $$

This function is monotonically non-decreasing, continuous, and the image is contained in $[0,\frac{1}{t+1}]$. We have $\beta<\frac 1t$, and it is in fact possible that
$\beta<\frac{t+1}{t(t+2)}$. It holds that $w(x)\leq x$ and $w(x)\geq \frac{tx}{t+1}$.

We compare the total weights of the items of a cluster and the total weights of items of one bin of an optimal solution.
Consider a bin of an optimal solution. Excluding the largest item of the bin, the total size is below $1$, and the weight of these items is also below $1$. The largest item has size at most $\beta$, and due to monotonicity, its weight is at most $w(\beta)$. The total weight is at most $1+w(\beta)$, which is $1+\frac {1}{t+2}=\frac{t+3}{t+2}$ if $\frac 1{t+1}\leq \beta \leq \frac{t+1}{t(t+2)}$, and it is at most $1+\beta\cdot \frac t{t+1}$ if
$\frac{t+1}{t(t+2)} \leq \beta < \frac 1t$.

To prove the upper bound, it is left to show that  $W_i \geq \frac{OPT_i}2$ holds for any cluster $i$ with $OPT_i \geq 2$ bins. We apply FF on cluster $i$ and obtain $A_i \geq OPT_i$ bins. Every bin out of the first $A_i-2 \geq 0$ bins has a total size of items of at least $1$, and therefore the weight of each such bin is at least $\frac{t}{t+1} \geq \frac 12$. Thus, we consider the last two bins and show that the total weight for this pair of bins is at least $1$. Consider the set $Q$ of all the items of the penultimate bin together with the first item of the last bin.
Let $\Gamma$ (where $0<\Gamma\leq \beta < \frac1t$) be the size of an item of maximum size of any item in $Q$.  We claim that the total size of items for the last two bins, excluding an item of size $\Gamma$, is at least $1$. In fact this holds for $Q$, since otherwise the first item of the last bin could be packed into the penultimate bin, since the total size is tested for all items of $Q$ excluding the largest item. Let $k$ denote the number of items with sizes at least $\frac{1}{t+2}$ in $Q$. Since the weight of each such item is at least $\frac{1}{t+2}$, we are done in the case $k\geq t+2$. If $k=0$, the weight of each item is equal to its size, and the total weight is above $1$. We are left with the case $1\leq k \leq t+1$, where $\Gamma \geq \frac{1}{t+2}$, since it is the largest item, and there is at least one item with size at least $\frac{1}{t+2}$.

Let $\rho\geq \Gamma$ be the total size of items of sizes at least $\frac{1}{t+2}$ in the set, excluding the largest such item. We have $\rho \leq (k-1)\cdot \Gamma \leq t \cdot \Gamma < t \cdot \frac 1t =1$, and therefore the total size of small items is at least $1-\rho >0$, which is also a lower bound on their total weight. For the other items we use $w(x)\geq \frac{tx}{t+1}$.
The total weight is at least $\frac{t}{t+1}\cdot (\rho+\Gamma)+(1-\rho)=\frac{t}{t+1}\cdot \Gamma +1 -\frac{\rho}{t+1} \geq 1$, since $\rho \leq t\cdot \Gamma$.
\end{proof}

\begin{lemma}\label{le14}
For $0<\beta<1$, the PoC of max-OEBP is at least $R_2(\beta)$ for items of sizes in $(0,\beta]$.
\end{lemma}
%\subsection{Proof of Lemma \ref{le14}}
\begin{proof}
We obtain the result using two types of constructions, for the two cases in the function $R$. %% Recall that $t=\lceil \frac{1}{\beta} \rceil-1$ and $\beta\in[\frac{1}{t+1},\frac 1t)$.

The first construction consists of items of sizes $\frac{1}{t+1}$ (large items) and $\eps$ (small items), where $\eps=\frac{1}{M}$ for a large integer $M$. The number of large items is $(t+2)\cdot(M+1)$, and the number of small items is $(t+2)\cdot (M^2-1)$. In an optimal solution for the entire input, there are $(t+2)\cdot(M+1)$ bins, where every bin has one large item and $M-1$ small items (where the total size of small items for each bin is $1-\frac 1M$).

There are two types of clusters. There are $(t+2)\cdot(M-1)$ clusters with $M+1$ small items each, and there are $M+1$ clusters with $t+2$ large items each. The total size of $M$ small items is $1$, and therefore every cluster of small items requires two bins. The total size of $t+1$ large items is also $1$, so every cluster of large items also requires two bins. Thus, the total number of bins is $2(t+2)(M-1)+2(M+1)=2(t+3)M-2(t+1)$. The ratio $\frac{2(t+3)M-2(t+2)}{(t+1)(M+1)}$ tends to $2\cdot\frac{t+3}{t+2}$ as $M$ grows to infinity.

The second construction consists of items of sizes $\beta$ (large items) and $\eps$ (small items), where $\eps=\frac{1}{M}$ for a large integer $M>3$. The number of large items is $(t+1)\cdot M$, and the number of small items is $(t+1)\cdot (M^2-M)$. In an optimal solution for the entire input, there are $(t+1)\cdot M$ bins, where every bin has one large item and $M-1$ small items (where the total size of small items for each bin is $1-\frac 1M$).

There are two types of clusters. There are $M$ clusters having $t+1$ large items each, and in addition each such cluster has $\lceil M(1-t\cdot \beta)\rceil$ small items. The total size of the small items of such a cluster is
$\lceil M(1-t\cdot \beta)\rceil \cdot \frac 1M \geq 1-t\cdot \beta$, and therefore the total size of the small items together with $t$ of the large items is at least $1$, and every cluster requires two bins.
Since $\beta\cdot t\geq \frac{t}{t+1}\geq \frac 12$, the number of small items in all such clusters together is at most $M\cdot \lceil M(1-t\cdot \beta)\rceil \leq M\cdot \lceil \frac M2 \rceil \leq \frac{M(M+1)}2<M^2-M$, by $M>3$. The number of small items is $(t+1)\cdot (M^2-M) \geq 2(M^2-M)$, so the number of remaining small items is above $M^2-M>M+1$. The remaining small items are split into clusters with $M+1$ items, and if there is a residue with less than $M+1$ items, it is added to one of the clusters of small items. Since $M$ small items have total size $1$, every cluster requires two bins (one such cluster may require three bins if it has $2M+1$ small items).

The number of clusters with only small items is $$\lfloor\frac{(t+1)(M^2-M)-M\cdot \lceil M(1-t\cdot \beta)\rceil}{M+1}\rfloor \geq \frac{(t+1)(M^2-M)-M\cdot ( M(1-t\cdot \beta)+1)}{M+1} -1 $$ $$=\frac{M^2(t(1+\beta))-M(t+3)-1}{M+1} \ , $$ and the total number of clusters is at least $\frac{M^2(t(1+\beta))-M(t+3)-1}{M+1}+M=\frac{M^2(1+t(1+\beta))-M(t+2)-1}{M+1}$, where every cluster has at least two bins.
The ratio $\frac{2(M^2(1+t(1+\beta))-M(t+2)-1)/(M+1)}{(t+1)M}$ tends to $2\cdot\frac{1+t+t\beta}{t+1}=2+2\cdot\frac{t\beta}{t+1}$ as $M$ grows to infinity.
\end{proof}

We conclude with the following theorem.
\begin{theorem}
The PoC of Max-OEBP is equal to $R_2(\beta)$ for the parameter $\beta$. In particular, it is equal to $3$ for the general case.
\end{theorem}

\section{Some results for Min-OEBP}\label{mino}
In this section we discuss the other variant, called Min-OEBP. In this case, FF is defined similarly to Max-OEBP, with the difference in the definition of validity. For example, if the input consists of items of sizes $0.9$, $0.4$, $0.4$, then the second item can be packed with the first item for both variants, but the third item cannot be added to the bin for Min-OEBP, while it can be added for Max-OEBP. This holds as the sum for Max-OEBP is calculated without the largest item and the sum for Min-OEBP is calculated without the smallest item (this sum has to be strictly below $1$ in both cases). For Min-OEBP it is possible that greedy algorithms will create bins with total size strictly below $1$ (see \cite{LYX10T}). In the example above, if the items of sizes $0.4$ arrive before the item of size $0.9$, then a new bin has to be opened.

Min-OEBP was studied with respect to greedy algorithms \cite{LYX10T}, not only for the general case but also for the parametric case with $\beta$ being a reciprocal of an integer. Generalizing to arbitrary values of $\beta$ seems to be hard and we do not deal with that (see some comments regarding this below). Here, we study the PoC for this problem, and see that the parameterized case is also hard for this problem too.
In the analysis we use NFD (which is the same algorithm as FFD for this problem). These algorithms are defined as NF and FF with validity defined for Min-OEBP, and sorted inputs. For NFD (and FFD), all bins possibly excluding the last bin will contain a total size of items of at least $1$ (unlike NF and FF). This holds since at each time during the execution of the algorithm the considered item is the smallest one, so it can be added to any bin whose load is strictly below $1$. It will always hold that the last item packed into the bin is in fact its smallest item.

The greedy algorithms analyzed by \cite{LYX10T} are NF, WF, and FF, for which the asymptotic approximation ratios are $4$, $3$, and $2$, respectively. For NFD (and FFD) the ratio is the same as for the harmonic algorithm $\frac{71}{60}$.

\begin{theorem}\label{tm16}
The PoC of Min-OEBP is equal to $4$ for the general case.
\end{theorem}
%\subsection{Proof of Theorem \ref{tm16}}
\begin{proof}
First, we show that For any input $I$ for min-OEBP, it holds that $\sum_{j=1}^{\ell} OPT_j \leq 4\cdot OPT$.
The proof of the upper bound is simple, and it is possible to use a similar proof to the analysis of Next Fit (NF) for Min-OEBP \cite{LYX10T}. We will prove the upper bound using NFD, for which the proof is even shorter.
%
%Consider a bin of an optimal solution. The total size of all items of any bin is not larger than $2$, since the total size excluding the largest item is below $1$. Let the total size of items be $X$. We have $OPT\geq \frac{X}{2}$.
%
%Consider a cluster for $I_j$, and an solution of FF for it, which contains $A_j\geq 2$ bins. For every bin that is not the first bin, its first item could not be added to any earlier bin, so the total size of the first item and a subset of items of the earlier bin is above $1$. The subset is always proper since the smallest item is excluded, but it is also possible that some items arrived later. Thus, the total size of items for this cluster, which we denote by $X_j$ is above $\frac{A_j}2$, and we have $\sum_{j=1}^{\ell} OPT_j \leq \sum_{j=1}^{\ell} A_j \leq \sum_{j=1}^{\ell} 2\cdot X_j = 2\cdot X \leq 4 \cdot OPT$.
%
%Another proof:
Use NFD for every cluster. Every bin except for possibly the last one has a total size no smaller than $1$, the last bin is not empty, and $A_j\geq 2$.
Thus, the total size of items for this cluster, which we denote by $X_j$ is above $A_j-1 \geq \frac{A_j}2$, and we have $\sum_{j=1}^{\ell} OPT_j \leq \sum_{j=1}^{\ell} A_j \leq \sum_{j=1}^{\ell} 2\cdot X_j = 2\cdot X \leq 4 \cdot OPT$.

Next, we show the lower bound, that the PoC of min-OEBP is at least $4$.
The proof of the lower bound is also related to the analysis of Next Fit (NF) for Min-OEBP \cite{LYX10T}, even though packing clusters using NF for the sake of proof would not allow us to prove the bound. Consider the example where every cluster has four items of sizes $\frac 2{N^2}, \frac 2{N^2}, 1-\frac 1{N^2}, \frac 2{N^2}, \frac 2{N^2}$ for a large positive integer $N$,  and the number of clusters is $2N$.

If NF is applied for each cluster, the third item is packed into a new bin, and the fifth item is also packed into a new bin.
An optimal solution for every cluster cannot use just one bin, but it does not require three bins. There is a different solution for each cluster uses two bins, one with the largest item, and one with the other items, and this is the optimal solution for every cluster.
Thus, the cost for all clusters is $4N$, while an optimal solution for all items packs the items of sizes $1-\frac 1{N^2}$ in pairs into $N$ bins, and all other items have a total size of $2N \cdot 4 \cdot \frac 2{N^2}=\frac 8N$, and they are packed into one bin. This example provides us a lower bound of $4$ on the PoC, since the PoC is equal to $\frac{4N}{N+1}$, which tends to $4$ as $N$ grows to infinity.
\end{proof}

\medskip

We use the next lemma in the analysis of the parametric case. We will consider two interesting special cases for the PoC. One case is where $\beta$ is a reciprocal of an integer, and the other case is when $1-\beta$ has this property. The cases $\beta=1-\frac 13$ and $\beta=1-\frac 14$ are proved using a complete analysis for $[0.5 \leq \beta \leq 0.8]$.

\begin{lemma}\label{le17}
The PoC of min-OEBP is and $\beta \in [\frac 1{t+1},\frac 1{t})$, where $t\geq 1$ is an integer, is at least $2\cdot \frac {t+2}{t+1}$, and it is at least $\frac{4(k+1)}{k+3}$ for $\beta= 1- \frac 1k$ for an integer $k\geq 5$.
\end{lemma}
\begin{proof}
For the first lower bound, let $N$ be a large integer and let the input consist of $(t+1)(t+2)N$ items of size $\frac{1}{t+1}-\frac{1}{(t+2)N}$ (large items). There are also $2(t+2)N$ items of size $\frac{1}{N}$ (small items). An optimal solution has $(t+1)\cdot N$ bins with $t+2$ large items. The bins are valid since the total size of $t+1$ large items is strictly below $1$. It also has $2(t+2)$ bins with $N$ small items each.

For the algorithm, every cluster has $t+1$ large items and two small items. Thus, there are $(t+2)N$ clusters. Every cluster requires two bins, since the total size of $t+1$ large items and one small item is $1-\frac{t+1}{(t+2)N}+\frac 1N=1+\frac{1}{(t+2)N}$ (and it is possible to use exactly two bins by packing large items  separately from small items). Thus, the cost of packing using clusters is $2(t+2)N$, and the ratio between the optimal costs with and without clusters is $\frac{2(t+2)N}{(t+1)\cdot N+2(t+2)}$. Letting $N$ grow without bound, the ratio is $\frac{2(t+2)}{t+1}$.

For the second lower bound, let $N$ be a large integer.  Let the input consist of $2(k+1)N$ items of size $\beta$ (large items), $2(k+1)N$  items of size $\frac 1k - \frac{1}{2N}$ (medium items), and $4(k+1)N$ items of size $\frac{1}{N}$ (small items). An optimal solution has $(k+1)N$ bins with two large items each, $2N$ bins with $k+1$ medium items each, and $4(k+1)$ bins with $N$ small items each.

For the algorithm, every cluster one large item, one medium item, and two small items. Thus, there are $2(k+1)N$ clusters. Every cluster requires two bins, since the total size of one large item, one medium item, and one small item is $(1-\frac 1k)+(\frac 1k-\frac 1{2N})+\frac 1N=1 +\frac 1{2N}$ (and it is possible to use exactly two bins by packing large items separately from small items). Thus, the cost of packing using clusters is $4(k+1)N$, and the ratio between the optimal costs with and without clusters is $\frac{4(k+1)N}{(k+1)N+2N+4(k+1)}$. Letting $N$ grow without bound, the ratio is $\frac{4(k+1)}{k+3}$.
\end{proof}

\medskip

Before continuing to upper bounds for the parametric case, we briefly discuss the known analysis of greedy heuristics for the parametric case with $\beta=\frac 1q$ for any integer $q\geq 1$  \cite{LYX10T}.
For NF and WF, in the case $\beta=\frac 1k$, the asymptotic approximation ratio is $\frac{k+1}{k-1}$. For FF it is $\frac{k+1}k$, and for NFD (and FFD) it is still $\frac{71}{60}$ for $k=2$, and it is $1+\frac{1}{k+2}-\frac{1}{k(k+1)(k+2)}$ for odd $k$ and $1+\frac{1}{k+2}-\frac{2}{k(k+1)(k+2)}$ for even $k$. Thus, the ratio is $\frac{71}{60}$ for $\beta\geq \frac 13$.
For WF the ratio is $3$ for $\beta=1$ and $\beta=\frac 12$, and therefore the ratio is $3$ for all $\frac 12 \leq \beta \leq 1$. The ratio is $2$ for $\beta=\frac 13$.

However, if we try to analyze intermediate values of $\beta$, such as $\beta=0.45$, an example with very small items, items of sizes just below $0.1$ and items of size $0.45$ (every bin of the algorithm has one item of size $0.45$, one item of size slightly below $0.1$ and a small number of very small items) gives a lower bound of $\frac{33}{14}\approx  2.35714$, since an optimal solution can pack identical items into every bin.  It is not difficult to prove an upper bound below $3$ for every $\frac 13< \beta < \frac 12$, since no bin can have a total size of items above $1+\beta$, and no bin of any AF algorithm that is not the last bin  can have a total size below $1-\beta$. For $\beta>\frac 5{12}$, even a simple example with items of size $\beta$ and very small items (so that every bin of the algorithm has one items of size $\beta$ and very small items of size slightly above $1-2\beta$) already gives a lower bound strictly above $2$.

\begin{theorem}\label{tm18}
The PoC of min-OEBP is  equal to $3$ for $\beta\in [\frac 12,0.8]$. It is equal to $\frac{4(k+1)}{k+3}$
for $\beta = 1- \frac 1k$ and any integer $k\geq 5$.  It is equal to $2\cdot \frac{t+2}{t+1}$ for $\beta=\frac 1{t+1}$ and an integer $t\geq 0$.
\end{theorem}
Note that the result indeed tends to $4$, which is the result for $\beta=1$, that is, for $k$ growing to infinity. The case $t=0$ is in fact the case $\beta=1$ analyzed earlier.

\medskip

\noindent\begin{proof}
The lower bounds follow from the previous lemma. Specifically, the lower bound for  $[\frac 12,0.8]$ follows from the case $\beta=\frac 12$.

For $\beta=\frac 1{t+1}$, the upper bound holds because every bin of any cluster is full by more than half on average, while no bin of an optimal solution can be full by $1+\beta$ or more. We will consider the case $\beta=1-\frac 1k$ for $k\geq 5$, where the upper bound for $k=5$ will imply the upper bound for $\beta\in [\frac 12,0.8]$.

Let $$w(x)=\begin{cases} \vspace{0.3cm}
  \frac{k+1}{k+3}  \ {  \mbox \it \  \  \ \ \  \ \ \ \ for  \ \ \  \ \ \ \ \ }  \ \ \  x  \in (\frac {k+1}{k+3},\frac{k-1}k=\beta]\\
 x  \ { \mbox \it \  \ \ \ \ \ \  \   \  \  \ \ \ for \ \ \  \ \ \  \ \ \ \ \ }   x  \in (\frac 2{k+3},\frac{k+1}{k+3}]\\
 \frac{2}{k+3}  \ \ \ \  \  \  \ \ \ \ { \mbox \it \  for  \ \ \ \ \ }  \ \ \ \ \  x  \in (\frac 1k,\frac 2{k+3}] \ \ \\
 \frac{2k}{k+3}\cdot x  { \mbox \it \   \  \  \ \ \ for \  \ \ \ \ \ \ \  \ \  \  } x  \in (0,\frac 1k]\\
\end{cases}$$ be a weight function, which we will use in the analysis. The function is continuous and monotonically non-decreasing.

Note that $\frac{k+1}{k+3} \geq \frac 12$ and  $\frac{2}{k+3} \leq \frac 12$ hold for $k\geq 1$, so $\frac 12$ belongs to the interval $(\frac 2{k+3},\frac{k+1}{k+3}]$. Thus, for $x \in (\frac 12,\beta]$, we have $w(x)=\min\{x,\frac{k+1}{k+3}\}$.

Consider a bin of an optimal solution, and let $y$ be its smallest item as well as the size of the smallest item. The total size of other items of the same bin is below $1$. If $y>\frac 12$, all the items of the bin are larger than $\frac 12$, and there are at most two such items. Thus, the total weight is at most $\frac{2(k+1)}{k+3}$. If $\frac{2}{k+3}\leq x \leq \frac 12$, no item has size below $\frac{2}{k+3}$, and no item has weight larger than its size. The total size of items is at most $\frac 32$, and the same holds for the total weight. We have  $\frac 32 \leq \frac{2(k+1)}{k+3}$ for $k\geq 5$, so the required total weight for the bin is not exceeded. Finally, in the case where $y\leq \frac{2}{k+3}$, the weight of an item is at most $\frac{2k}{k+3}$ times its size and no item has weight above $\frac{2}{k+3}$, and the total weight for the bin is at most $\frac{2k}{k+3}+\frac{2}{k+3}=\frac{2(k+1)}{k+3}$.

Consider a cluster, and apply NFD on its items. Let $A_i$ be the number of resulting bins, where the optimal number of bins is at most $A_i$. Recall that the total size for every bin, except for possibly the last one, is at least $1$.
%Since for NFD the smallest item of every bin is its last item, the total size for every bin, except for possibly the last one, is at least $1$.

We will show that the total weight will always be above $1$ for a set of items whose total size is at least $1$. If there are no items with sizes above $\frac{k+1}{k+3}$, no item has weight smaller than its size (since  $\frac{2(k+1)}{k+3} \geq 1$) and the claim follows. If there are at least two such items, the total weight is at least $\frac{2(k+1)}{k+3}>1$. If there is exactly one such item, the remaining items have total size of at least $1-\beta=\frac 1k$. If there is an item of size above $\frac 1k$, its weight is at least $\frac{2}{k+3}$ (by the monotonicity of $w$), and the total weight is at least $\frac{k+1}{k+3}+\frac{2}{k+3}=1$. Otherwise, the total weight of remaining items is at least $\frac{2k}{k+3}\cdot \frac 1k=\frac{2}{k+3}$, and once again the total weight is at least $\frac{k+1}{k+3}+\frac{2}{k+3}=1$.

Since $A_i-1 \geq \frac{A_i}2$ holds for $A_i \geq 2$, the total weight for the cluster is at least $\frac{A_i}2$, and it is at least half of the number of all bins for all cluster in total.
\end{proof}

Other cases of $\beta$ are more difficult.
For example, for $\beta=1-\frac 16-\frac 1{42} =\frac{17}{21}\approx 0.8095238$, one can use items of sizes close to $\frac 16$ and $\frac 1{42}$ instead of items of sizes close to $\frac 15$ to obtain a construction where in an optimal solution there are bins with seven items of sizes close to $\frac 16$, and bins with $43$ items of sizes close to $\frac 1{42}$. This indeed gives a lower bound of approximately $3.002493765586$ on the PoC.
It may seem that for $\beta=1-\frac 16-\frac 1{36}\approx 0.80556$ we can also get a result larger than $3$, but this construction is inferior to the one giving a lower bound of $3$.

% for 0.81 we use 1/6, 1/43, 1/12900 to get 0.66566

\section{Conclusion}
We have analyzed greedy algorithms for Max-OEBP, and the price of clustering for Max-OEBP and for Min-OEBP.  These two problems have asymptotic polynomial time approximation schemes \cite{LDY01,EL08,LYX10T,BEL20,EL20}, and so do the other two variants. However, there are gaps in the study of online algorithms \cite{LDY01,EL20,BEL20}. Another interesting research direction is the study of additional variants of bin packing with respect to the type of analysis given here.

%\newpage

%\bibliography{pocv}

%\appendix

%\section{Omitted proofs}

%\medskip

\bibliographystyle{abbrv}

\bibliography{pocv}

\end{document}